\documentclass[runningheads]{llncs}
\usepackage[T1]{fontenc}
\usepackage{graphicx}
\newcommand\myoverset[2]{\overset{\textstyle #1\mathstrut}{#2}}

\newtheorem{innercustomlem}{Lemma}
\newenvironment{customlem}[1]
  {\renewcommand\theinnercustomlem{#1}\innercustomlem}
  {\endinnercustomlem}
  
  \newtheorem{innercustomthm}{Theorem}
\newenvironment{customthm}[1]
  {\renewcommand\theinnercustomthm{#1}\innercustomthm}
  {\endinnercustomthm}

\newtheorem{innercustomcor}{Corollary}
\newenvironment{customcor}[1]
  {\renewcommand\theinnercustomcor{#1}\innercustomcor}
  {\endinnercustomcor}
  
\newtheorem{innercustomprop}{Proposition}
\newenvironment{customprop}[1]
  {\renewcommand\theinnercustomprop{#1}\innercustomprop}
  {\endinnercustomprop}

\usepackage{amssymb, amsmath}
\usepackage{synttree, bussproofs,scrextend, stmaryrd, rotating, xcolor}
\usepackage{esvect,pgf, tikz}
\usetikzlibrary{arrows, automata, positioning}
\usepackage{comment}
\usepackage[all]{xy}
\usepackage{changepage,enumerate,proof,upgreek}
\usepackage{cleveref}
\usepackage{multicol}
\usepackage{mathtools}

\renewcommand{\vec}[1]{\mathbf #1}

\def\phi{\varphi}

\newcommand{\sect}{Section} 
\newcommand{\dfn}{Definition} 
\newcommand{\lem}{Lemma} 
\newcommand{\thm}{Theorem} 
\newcommand{\prp}{Proposition} 
\newcommand{\fig}{Figure} 
\newcommand{\iffi}{\textit{iff} } 

\definecolor{tim}{RGB}{0, 0, 250}

\newcommand{\ch}{\mathbf{Ch}}
\newcommand{\chk}[2]{\mathbf{Ch}_{#1}(#2)}

\newcommand{\ro}{(\mathsf{r})}
\newcommand{\rop}[1]{(\mathsf{r}_{#1})}

\newcommand{\at}{\mathrm{At}}
\newcommand{\lbl}{\mathrm{L}}
\newcommand{\nodes}{\mathrm{V}}
\newcommand{\edges}{\mathrm{E}}
\newcommand{\dg}{G_{\der}}

\newcommand{\os}{\Upsigma}
\newcommand{\empstr}{\varepsilon}
\newcommand{\CR}{(\mathsf{cr})}

\newcommand{\TR}{(\mathsf{tr})}
\newcommand{\AR}{(\mathsf{ar})}

\newcommand{\RO}{(\mathsf{r})}
\newcommand{\ilt}{<}

\newcommand{\N}{\vec{n}}

\newcommand{\subos}{\sqsubseteq}
\newcommand{\dlen}[1]{|#1|}
\newcommand{\vn}{n}
\newcommand{\vm}{m}
\newcommand{\vk}{k}

\newcommand{\der}{\delta}
 \newcommand{\drel}[2]{\mathop{\smash{\stackrel{\phantom{\underline{g}}{#1}\phantom{\underline{g}}}{\smash{\longrightarrow}}_{#2}}}} 

\newcommand{\ruleset}{\mathcal{R}}

\newcommand{\head}{\mathit{head}}
\newcommand{\body}{\mathit{body}}

\newcommand{\gbts}{\mathbf{gbts}}

\newcommand{\bts}{\mathbf{bts}}
\newcommand{\sbts}{\mathbf{fts}}

\newcommand{\adgs}{\mathbf{cdgs}}
\newcommand{\wadgs}{\mathbf{wcdgs}}
\newcommand{\cdgs}{\mathbf{cdgs}}
\newcommand{\wcdgs}{\mathbf{wcdgs}}
\newcommand{\wgbts}{\mathbf{wgbts}}
\newcommand{\wandgbts}{\mathbf{(w)gbts}}
\newcommand{\wandcdgs}{\mathbf{(w)cdgs}}
\newcommand{\applic}{\tau}

\newcommand{\rsii}{\mathcal{R}_{2}}
\newcommand{\rsiii}{\mathcal{R}_{3}}
\newcommand{\dbii}{\mathcal{D}_{\dag}}
\newcommand{\dbiii}{\mathcal{D}_{\ddag}}

\newcommand{\conset}{\mathbf{Con}} 
\newcommand{\cons}{C}
\newcommand{\nonc}{\overline{C}}
\newcommand{\nullset}{\mathbf{Nul}} 
\newcommand{\termset}{\mathbf{Ter}} 
\newcommand{\varset}{\mathbf{Var}} 
\newcommand{\con}[1]{#1}
\newcommand{\var}[1]{#1}

\newcommand{\cona}{\con{a}}
\newcommand{\conb}{\con{b}}
\newcommand{\conc}{\con{c}}

\newcommand{\varx}{\var{x}}
\newcommand{\vary}{\var{y}}
\newcommand{\varz}{\var{z}}
\newcommand{\sig}{\Sigma}

\newcommand{\arf}{ar}
\newcommand{\ins}{\mathcal{I}}

\newcommand{\db}{\mathcal{D}}
\newcommand{\atseti}{\mathcal{X}}
\newcommand{\atsetii}{\mathcal{Y}}
\newcommand{\imp}{\rightarrow}
\newcommand{\kb}{\mathcal{K}}
\newcommand{\rs}{\mathcal{R}}

\newcommand{\calr}{\mathcal{R}}

\newcommand{\fr}{\mathit{fr}}

\newcommand{\grd}[1]{G(#1)}
\newcommand{\grdv}{V}
\newcommand{\grde}{E}

\newcommand{\td}{T}
\newcommand{\tdv}{V}
\newcommand{\tde}{E}
\newcommand{\w}[1]{w(#1)}
\newcommand{\tw}[1]{tw(#1)}
\newcommand{\maxofset}[1]{\max\{#1\}}
\newcommand{\minofset}[1]{\min\{#1\}}

\begin{document}
%
\title{Derivation-Graph-Based Characterizations of Decidable Existential Rule Sets\thanks{Work supported by the ERC through Consolidator Grant 771779 
 (DeciGUT).}}
%
%
\author{Tim S. Lyon\orcidID{0000-0003-3214-0828} \and
Sebastian Rudolph\orcidID{0000-0002-1609-2080}} 
\authorrunning{T. S. Lyon \and S. Rudolph}
%
\institute{Computational Logic Group, TU Dresden, Germany
\email{\{timothy\_stephen.lyon,sebastian.rudolph\}@tu-dresden.de}}
\maketitle              
\begin{abstract}
This paper establishes alternative characterizations of very expressive classes of existential rule sets with decidable query entailment. We consider the notable class of greedy bounded-treewidth sets ($\gbts$) and a new, generalized variant, called weakly gbts ($\wgbts$). Revisiting and building on the notion of derivation graphs, we define (weakly) cycle-free derivation graph sets ($\wandcdgs$) and employ elaborate proof-theoretic arguments to obtain that $\gbts$ and $\cdgs$ coincide, as do $\wgbts$ and $\wcdgs$. These novel characterizations advance our analytic proof-theoretic understanding of existential rules and will likely be instrumental in practice.
\end{abstract}
%
%
%


\section{Introduction}


 The formalism of existential rules has come to prominence as an effective approach for both specifying and querying knowledge. Within this context, a knowledge base takes the form $\mathcal{K} = (\db,\ruleset)$, where $\db$ is a finite collection of atomic facts (called a \emph{database}) and $\ruleset$ is a finite set of \emph{existential rules} (called a \emph{rule set}), which are first-order formulae of the form $\forall \vec{x} \vec{y} (\phi(\vec{x},\vec{y}) \imp \exists \vec{z} \psi(\vec{y},\vec{z}))$. Although existential rules are written in a relatively simple language, they are expressive enough to generalize many important languages used in knowledge representation, including rule-based formalisms as well as such based on description logics. Moreover, existential rules have meaningful applications within the domain of ontology-based query answering~\cite{BagLecMugSal09}, data exchange and integration~\cite{FagKolMilPop05}, and have proven beneficial in the study of general decidability criteria~\cite{FelLyoOstRud23}.

\begin{figure}[b]\label{fig:ER-classes-and-inclusions}
~\hfill\begin{tabular}{c}
\begin{tikzpicture}
\pgftransformscale{.65}

\draw[] (-4.5,0) parabola bend (0,4.5) (4.5,0);
\node at (0,4) {$\bts$};



\draw[gray] (-3.5,0) parabola bend (0,3.5) (3.5,0);
\node 
at (0,2.4) {$\wgbts = \wadgs$};

\draw[gray] (-2,0) parabola bend (0,2) (2,0);
\node at (0,.75) {$\gbts = \adgs$};

\draw[very thick] (4.5,0) -- (-4.5,0);
\end{tikzpicture}
\end{tabular}\hfill~

\caption{A graphic depicting the containment relations between the classes of rule sets considered. The solid edges represent strict containment relations.}
\end{figure}
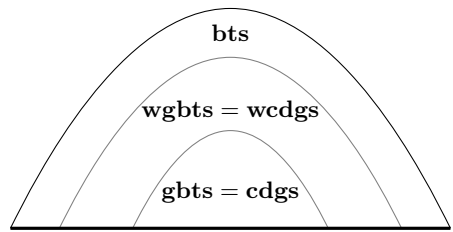

The \emph{Boolean conjunctive query entailment problem} consists of taking a knowledge base $\kb$, a Boolean conjunctive query (BCQ) $q$, and determining if $\kb \models q$. As this problem is known to be undecidable for arbitrary rule sets~\cite{ChandraLM81}, much work has gone into identifying existential rule fragments for which decidability can be reclaimed. Typically, such classes of rule sets are described in one of two ways: either, a rule set's membership in said class can be established through easily verifiable \emph{syntactic properties} (such classes are called \emph{concrete classes}), or the property is more \emph{abstract} (which is often defined on the basis of semantic notions) and may be hard or even impossible to algorithmically determine (such classes are called \emph{abstract classes}). Examples of concrete classes include functional/inclusion dependencies~\cite{JohKlu82}, datalog, and guarded rules~\cite{CaliGK13}. Examples of abstract classes include finite expansion sets~\cite{BagMug02}, finite unification sets~\cite{BagLecMug09}, and bounded-treewidth sets ($\bts$)~\cite{CaliGK13}.

Yet, there is another means of establishing the decidability of query entailment: only limited work has gone into identifying classes of rule sets with decidable query entailment based on their \emph{proof-theoretic characteristics}, in particular, based on specifics of the derivations such rules produce. To the best of our knowledge, only the class of \emph{greedy bounded treewidth sets} ($\gbts$) has been identified in such a manner (see~\cite{ThoBagMugRud12}). A rule set qualifies as $\gbts$ when every derivation it produces is \emph{greedy}, in a sense that it is possible to construct a tree decomposition of finite width in a ``greedy'' fashion alongside the derivation, ensuring the existence of a model with finite treewidth for the knowledge base under consideration, thus warranting the decidability of query entailment~\cite{CaliGK13}. 

In this paper, we investigate the $\gbts$ class and three new classes of rule sets where decidability is determined proof-theoretically. First, we define a weakened version of $\gbts$, dubbed $\wgbts$, where the rule set need only produce \emph{at least one greedy derivation} relative to any given database. Second, we investigate two new classes of rule sets, dubbed \emph{cycle-free derivation graph sets} ($\adgs$) and \emph{weakly cycle-free derivation graph sets} ($\wadgs$), which are defined relative to the notion of a \emph{derivation graph}. Derivation graphs were introduced by Baget et al.~\cite{BagLecMugSal11} and are directed acyclic graphs encoding \emph{how} certain facts are derived in the course of a derivation. Notably, via the application of \emph{reduction operations}, a derivation graph may be reduced to a tree, which serves as a tree decomposition of a model of the considered knowledge base. Such objects helped establish that (weakly) frontier-guarded rule sets are $\bts$~\cite{BagLecMugSal11}. In short, our key contributions are: 
\begin{enumerate}
\item We investigate how proof-theoretic structures gives rise to decidable query entailment and propose three new classes of rule sets. 
\item We show that $\gbts = \adgs$ and $\wgbts = \wadgs$, establishing a correspondence between greedy derivations and reducible derivation graphs. 
\item 
We show that $\wgbts$ \emph{properly subsumes} $\gbts$ via a novel proof transformation argument. Therefore, by the former point, we also find that $\wadgs$ properly subsumes $\adgs$. 

\end{enumerate}

The paper is organized accordingly: In \sect~\ref{sec:prelims}, we define preliminary notions. 
 We study $\gbts$ and $\wgbts$ in \sect~\ref{sec:greedy}, and show that the latter class properly subsumes the former via an intricate proof transformation argument. In \sect~\ref{sec:derivation-graphs}, we define $\adgs$ and $\wadgs$ as well as show that $\gbts = \adgs$ and $\wgbts = \wadgs$. Last, in \sect~\ref{sec:conclusion}, we conclude and discuss future work.

\section{Preliminaries}\label{sec:prelims}

\medskip \noindent {\bf Syntax and formulae.} We let $\termset$ be a set of \emph{terms}, which is the the union of three countably infinite, pairwise disjoint sets, namely, the set of \emph{constants} $\conset$, the set of \emph{variables} $\varset$, and the set of \emph{nulls} $\nullset$. We use $\cona$, $\conb$, $\conc$, $\ldots$ (occasionally annotated) to denote constants, and $\varx$, $\vary$, $\varz$, $\ldots$ (occasionally annotated) to denote both variables and nulls. A \emph{signature} $\sig$ is a set of \emph{predicates} $p$, $q$, $r$, $\ldots$ (which may be annotated) such that for each $p \in \sig$, $\arf(p) \in \mathbb{N}$ is the \emph{arity} of $p$. 
 For simplicity, we assume a fixed signature $\sig$ throughout the paper.

An \emph{atom} over $\sig$ is defined to be a formula of the form $p(t_{1}, \ldots, t_{n})$, where $p \in \sig$, $\arf(p) = n$, and $t_{i} \in \termset$ for each $i \in \{1, \ldots, n\}$. A \emph{ground atom} over $\sig$ is an atom $p(a_{1}, \ldots, a_{n})$ such that $a_{i} \in \conset$ for each $i \in \{1, \ldots, n\}$. We will often use $\vec{t}$ to denote a tuple $(t_{1}, \ldots, t_{n})$ of terms and $p(\vec{t})$ to denote a (ground) atom $p(t_{1}, \ldots, t_{n})$. An \emph{instance} over $\sig$ is defined to be a (potentially infinite) set $\ins$ of atoms over constants and nulls, and a \emph{database} $\db$ is a finite set of ground atoms. We let $\atseti$, $\atsetii$, $\ldots$ (occasionally annotated) denote (potentially infinite) sets of atoms with $\termset(\atseti)$, $\conset(\atseti)$, $\varset(\atseti)$, and $\nullset(\atseti)$ denoting the set of terms, constants, variables, and nulls occurring in the atoms of $\atseti$, respectively.

\medskip \noindent {\bf Substitutions and homomorphisms.} A \emph{substitution} is a partial function over the set of terms $\termset$. A \emph{homomorphism} $h$ from a set $\atseti$ of atoms to a set $\atsetii$ of atoms, is a substitution $h : \termset(\atseti) \to \termset(\atsetii)$ such that (i) $p(h(t_{1}), \ldots, h(t_{n})) \in \atsetii$, if $p(t_{1}, \ldots, t_{n}) \in \atseti$, and (ii) $h(a) = a$ for each $a \in \conset$. If $h$ is a homomorphism from $\atseti$ to $\atsetii$, we say that $h$ \emph{homomorphically maps} $\atseti$ to $\atsetii$. Atom sets $\atseti,\atsetii$ are \emph{homomorphically equivalent}, written $\atseti \equiv \atsetii$, \iffi $\atseti$ homomorphically maps to $\atsetii$, and vice versa. An \emph{isomorphism} is a bijective homomorphism $h$ where $h^{-1}$ is a homomorphism.

\medskip \noindent {\bf Existential rules.} Whereas databases encode assertional knowledge, ontologies consist in the current setting of \emph{existential rules}, which we will frequently refer to as \emph{rules} more simply. An existential rule is a first-order sentence of the form:
$$
\rho = \forall \vec{x} \vec{y} (\phi(\vec{x},\vec{y}) \imp \exists \vec{z} \psi(\vec{y},\vec{z}))
$$
where $\vec{x}$, $\vec{y}$, and $\vec{z}$ are pairwise disjoint collections of variables, $\phi(\vec{x},\vec{y})$ is a conjunction of atoms over constants and the variables $\vec{x},\vec{y}$, and $\psi(\vec{y},\vec{z})$ is a conjunction of atoms over constants and the variables $\vec{y},\vec{z}$. We define $\body(\rho) = \phi(\vec{x},\vec{y})$ to be the \emph{body} of $\rho$, and $\head(\rho) = \psi(\vec{y},\vec{z})$ to be the \emph{head} of $\rho$. For convenience, we will often interpret a conjunction $p_{1}(\vec{t_{1}}) \land \cdots \land p_{n}(\vec{t_{n}})$ of atoms (such as the body or head of a rule) as a set $\{p_{1}(\vec{t_{1}}), \cdots, p_{n}(\vec{t_{n}})\}$ of atoms; if $h$ is a homomorphism, then $h(p_{1}(\vec{t_{1}}) \land \cdots \land p_{n}(\vec{t_{n}})) := \{p_{1}(h(\vec{t_{1}})), \cdots, p_{n}(h(\vec{t_{n}}))\}$ with $h$ applied componentwise to each tuple $\vec{t_{i}}$ of terms. The \emph{frontier} of $\rho$, written $\fr(\rho)$, is the set of variables $\vec{y}$ that the body and head of $\rho$ have in common, that is, $\fr(\rho) = \varset(\body(\rho)) \cap \varset(\head(\rho))$. We define a \emph{frontier atom} in a rule $\rho$ to be an atom containing at least one frontier variable. We use $\rho$ and annotated versions thereof to denote rules, as well as $\ruleset$ and annotated versions thereof to denote finite sets of rules (simply called \emph{rule sets}).

\medskip \noindent {\bf Models.} We note that sets of atoms (which include instances and databases) may be seen as first-order interpretations, and so, we may use $\models$ to represent the satisfaction of formulae on such structures. A set of atoms $\atseti$ satisfies a set of atoms $\atsetii$ (or, equivalently, $\atseti$ is a model of $\atsetii$), written $\atseti \models \atsetii$, \iffi there exists a homomorphic mapping from $\atsetii$ to $\atseti$. A set of atoms $\atseti$ satisfies a rule $\rho$ (or, equivalently, $\atseti$ is a model of $\rho$), written $\atseti \models \rho$, \iffi for any homomorphism $h$, if $h$ is a homomorphism from $\body(\rho)$ to $\atseti$, then it can be extended to a homomorphism $\overline{h}$ that also maps $\head(\rho)$ to $\atseti$. A set of atoms $\atseti$ satisfies a rule set $\rs$ (or, equivalently, $\atseti$ is a model of $\rs$), written $\atseti \models \rs$, \iffi $\atseti \models \rho$ for every rule $\rho \in \rs$. If a  model $\atseti$ of a set of atoms, a rule, or a rule set homomorphically maps into \emph{every} model of that very set of atoms, rule, or rule set, then we refer to $\atseti$ as a \emph{universal model} of that set of atoms, rule, or rule set~\cite{DeuNasRem08}.

\medskip \noindent {\bf Knowledge bases and querying.} A \emph{knowledge base (KB)} $\kb$ is defined to be a pair $(\db,\rs)$, where $\db$ is a database and $\rs$ is a rule set. An instance $\ins$ is a \emph{model} of $\kb = (\db,\rs)$ \iffi $\db \subseteq \ins$ and $\ins \models \rs$. We consider querying knowledge bases with \emph{conjunctive queries (CQs)}, that is, with formulae of the form $q(\vec{y}) = \exists \vec{x} \phi(\vec{x},\vec{y})$, where $\phi(\vec{x},\vec{y})$ is a non-empty conjunction of atoms over the variables $\vec{x},\vec{y}$ and constants. We refer to the variables $\vec{y}$ in $q(\vec{y})$ as \emph{free} and define a \emph{Boolean conjunctive query (BCQ)} to be a CQ without free variables, i.e. a BCQ is a CQ of the form $q = \exists \vec{x} \phi(\vec{x})$. A knowledge base $\kb = (\db, \rs)$ \emph{entails} a CQ $q(\vec{y})= \exists \vec{x} \phi(\vec{x},\vec{y})$, written $\kb \models q(\vec{y})$, \iffi $\phi(\vec{x},\vec{y})$ homomorphically maps into every model $\ins$ of $\kb$; we note that this is equivalent to $\phi(\vec{x},\vec{y})$ homomorphically mapping into a universal model of $\db$ and $\calr$.

As we are interested in extracting implicit knowledge from the explicit knowledge presented in a knowledge base $\kb = (\db,\rs)$, we are interested in deciding the \emph{BCQ entailment problem}:\footnote{We recall that entailment of non-Boolean CQs or even query answering can all be reduced to BCQ entailment in logarithmic space.}
\begin{flushleft}
(BCQ Entailment) Given a KB $\kb$ and a BCQ $q$, is it the case that $\kb \models q$?
\end{flushleft}
While it is well-known that the BCQ entailment problem is undecidable in general~\cite{ChandraLM81}, restricting oneself to certain classes of rule sets (e.g. datalog or finite unification sets~\cite{BagLecMugSal11}) may recover decidability. We refer to classes of rule sets for which BCQ entailment is decidable as \emph{query-decidable classes}.

\medskip \noindent {\bf Derivations.} One means by which we can extract implicit knowledge from a given KB is through the use of \emph{derivations}, that is, sequences of instances obtained by sequentially applying rules to given data. We say that a rule $\rho = \forall \vec{x} \vec{y} (\phi(\vec{x},\vec{y}) \imp \exists \vec{z} \psi(\vec{y},\vec{z}))$ is \emph{triggered} in an instance $\ins$ via a homomorphism $h$, written succinctly as $\applic(\rho,\ins,h)$, \iffi $h$ homomorphically maps $\phi(\vec{x},\vec{y})$ to $\ins$. 
In this case, we define $$\ch(\ins,\rho,h) = \ins \cup \overline{h}(\psi(\vec{y},\vec{z}))$$, where $\overline{h}$ is an extension of $h$ mapping every variable $z$ in $\vec{z}$ to a fresh null.
Consequently, we define an \emph{$\rs$-derivation} to be a sequence $\ins_{0}, (\rho_{1}, h_{1}, \ins_{1}), \ldots, (\rho_{n}, h_{n}, \ins_{n})$ such that (i) $\rho_{i} \in \rs$ for each $i \in \{1, \ldots, n\}$, (ii) $\applic(\rho_{i},\ins_{i-1},h_{i})$ holds for $i \in \{1, \ldots, n\}$, and (iii) $\ins_{i} = \ch(\ins_{i-1},\rho,h_{i})$ for $i \in \{1, \ldots, n\}$. We will use $\der$ and annotations thereof to denote $\rs$-derivations, and we define the length of an $\rs$-derivation $\der = \ins_{0}, (\rho_{1}, h_{1}, \ins_{1}), \ldots, (\rho_{n}, h_{n}, \ins_{n})$, denoted $\dlen{\der}$, to be $n$. Furthermore, for instances $\ins$ and $\ins'$, we write $\ins \drel{\der}{\rs} \ins'$ to mean that there exists an $\rs$-derivation $\der$ of $\ins'$ from $\ins$. Also, if $\ins''$ can be derived from $\ins'$ by means of a rule $\rho \in \rs$ and homomorphism $h$, we abuse notation and write $\ins \drel{\der}{\rs} \ins', (\rho,h,\ins'')$ to indicate that $\ins \drel{\der}{\rs} \ins'$ and $\ins' \drel{\der'}{\rs} \ins''$ with $\der' = \ins', (\rho,h,\ins'')$. Derivations play a fundamental role in this paper as we aim to identify (and analyze the relationships between) query-decidable classes of rule sets based on \emph{how} such rule sets derive information, i.e. we are interested in classes of rule sets that may be \emph{proof-theoretically characterized}. 

\medskip \noindent {\bf Chase.} A tool that will prove useful in the current work is the \emph{chase}, which in our setting is a procedure that (in essence) simultaneously constructs all $\kb$-derivations in a breadth-first manner. Although many variants of the chase exist~\cite{BagLecMugSal11,FagKolMilPop05,MaiMenSag79}, we utilize the chase procedure (also called the \emph{k-Saturation}) from Baget et al.~\cite{BagLecMugSal11}. We use the chase in the current work as a purely technical tool for obtaining universal models of knowledge bases, proving useful in separating certain query-decidable classes of rule sets.

We define the \emph{one-step application} of all triggered rules from some $\rs$ in $\ins$ by
$$
\ch_{1}(\ins,\rs) = {\bigcup}_{\rho \in \rs, \applic(\rho,\ins,h)} \ch(\ins,\rho,h),
$$
assuming all nulls introduced in the ``parallel'' applications of $\ch$ to $\ins$ are distinct.
 We let $\ch_{0}(\ins,\rs) = \ins$, as well as let $\ch_{i+1}(\ins,\rs) = \ch_{1}(\ch_{i}(\ins,\rs),\rs)$, and define the \emph{chase} to be
$$
\ch_{\infty}(\ins,\rs) = {\bigcup}_{i \in \mathbb{N}} \ch_{i}(\ins,\rs).
$$
 For any KB $\kb = (\db,\rs)$, the chase $\ch_{\infty}(\db,\rs)$ is a universal model of $\kb$, that is, $\db \subseteq \ch_{\infty}(\db,\rs)$, $\ch_{\infty}(\db,\rs) \models \rs$, and $\ch_{\infty}(\db,\rs)$ homomorphically maps into every model of $\db$ and $\rs$.

\medskip \noindent {\bf Rule dependence.} Let $\rho$ and $\rho'$ be rules. We say that $\rho'$ \emph{depends} on $\rho$ \iffi there exists an instance $\ins$ such that (i) $\rho'$ is not triggered in $\ins$ via any homomorphism, (ii) $\rho$ is triggered in $\ins$ via a homomorphism $h$, and (iii) $\rho'$ is triggered in $\ch(\ins,\rho,h)$ via a homomorphism $h'$. We define the \emph{graph of rule dependencies}~\cite{Bag04} of a set $\rs$ of rules to be $\grd{\rs} = (\grdv,\grde)$ such that (i) $\grdv = \rs$ and (ii) $(\rho,\rho') \in \grde$ \iffi $\rho'$ depends on $\rho$. 

\medskip \noindent {\bf Treewidth.} A \emph{tree decomposition} of an instance $\ins$ is defined to be a tree $\td = (\tdv,\tde)$ such that $V \subseteq 2^{\termset(\ins)}$ (where each element of $V$ is called a \emph{bag}) and $E \subseteq V \times V$, satisfying the following three conditions: (i) $\bigcup_{X \in V} X = \termset(\ins)$, (ii)~for each $p(t_{1}, \ldots,t_{n}) \in \ins$, there is an $X \in V$ such that $\{ t_{1}, \ldots,t_{n} \}\subseteq X$, and (iii) for each $t \in \termset(\ins)$, the subgraph of $\td$ induced by the bags $X \in V$ with $t \in X$ is connected (this condition is referred to as the \emph{connectedness condition}). We define the \emph{width} of a tree decomposition $\td = (\tdv,\tde)$ of an instance $\ins$ as follows:
$$
\w{\td} := \maxofset{|X| : X \in V} - 1
$$
i.e.~the width is equal to the cardinality of the largest node in $\td$ minus 1. We let $\w{\td} = \infty$ \iffi for all $n \in \mathbb{N}$, $n \leq \maxofset{|X| : X \in V}$. We define the \emph{treewidth} of an instance $\ins$, written $\tw{\ins}$, as follows:
$$
\tw{\ins} := \minofset{\w{\td} : \text{ $\td$ is a tree decomposition of $\ins$}}
$$
i.e.~the treewidth of an instance equals the minimal width among all its tree de\-com\-positions. If no tree decomposition of $\ins$ has finite width, we set $\tw{\ins}=\infty$. 



\section{Greediness}\label{sec:greedy}

 We now discuss a property of derivations referred to as \emph{greediness}. In essence, a derivation is greedy when the image of the frontier of any applied rule consists solely of constants from a given KB and/or nulls introduced by a \emph{single} previous rule application. Such derivations were defined by Thomazo et al.~\cite{ThoBagMugRud12} and were used to identify the (query-decidable) class of \emph{greedy bounded-treewidth sets} ($\gbts$), that is, the class of rule sets that produce only \emph{greedy derivations} (defined below) when applied to a database.
 
 In this section, we also identify a new query-decidable class of rule sets, referred to as \emph{weakly greedy bounded-treewidth sets} ($\wgbts$). The $\wgbts$ class serves as a more liberal version of $\gbts$, and contains rule sets that admit at least one greedy derivation of any derivable instance. It is straightforward to confirm that $\wgbts$ generalizes $\gbts$ since if a rule set is $\gbts$, then every derivation of a derivable instance is greedy, implying that every derivable instance has \emph{some} greedy derivation. Yet, what is non-trivial to show is that $\wgbts$ \emph{properly subsumes} $\gbts$. We are going to prove this fact by means of a proof-theoretic argument and counter-example along the following lines: first, we show under what conditions we can permute rule applications in a given derivation (see \lem~\ref{lem:permutation-lemma} below), and second, we provide a rule set which exhibits non-greedy derivations (witnessing that the rule set is not $\gbts$), but for which every derivation can be transformed into a greedy derivation by means of rule permutations and replacements (witnessing $\wgbts$ membership).
 
 Let us now formally define greedy derivations and provide examples to demonstrate the concept of (non-)greediness. Based on this, we then proceed to define the $\gbts$ and $\wgbts$ classes.
 

\begin{definition}[Greedy Derivation~\cite{ThoBagMugRud12}]\label{def:greedy-derivation} We define an $\rs$-derivation
$$
\der = \ins_{0},(\rho_{1},h_{1},\ins_{1}), \ldots, (\rho_{n},h_{n},\ins_{n})
$$
to be \emph{greedy} \iffi for each $i$ such that $0 < i \leq n$, there exists a $j < i$ such that $h_{i}(\fr(\rho_{i})) \subseteq \nullset(\overline{h}_{j}(\head(\rho_{j}))) \cup \conset(\ins_{0},\rs) \cup \nullset(\ins_{0})$.
\end{definition}

 To give examples of non-greedy and greedy derivations, let us define the database $\dbii := \{p(a),r(b)\}$ and the rule set $\rsii := \{\rho_{1}, \rho_{2}, \rho_{3}, \rho_{4}\}$, with
$$
\begin{array}{rl@{\qquad}rl}
\rho_{1} = & \, p(x) \rightarrow \exists yz. q(x,y,z)     &    \rho_{3} = & \, p(x) \land r(y) \rightarrow \exists zwuv. q(x,z,w) \land s(y,u,v) \\[0.5ex]
\rho_{2} = & \, r(x) \rightarrow \exists yz. s(x,y,z)     &    \rho_{4} = & \, q(x,y,z) \land s(w,u,v) \rightarrow \exists o. t(x,y,w,u,o) \\
\end{array}
$$
 An example of a non-greedy derivation is the following:
$$
\der_{1} = \dbii, (\rho_{1},h_{1},\ins_{1}), (\rho_{1},h_{2},\ins_{2}), (\rho_{2},h_{3},\ins_{3}), (\rho_{4},h_{4},\ins_{4}),\text{\ \  with }
$$
$$
\ins_4 = \{ \underbrace{p(a),r(b)}_{\smash{\dbii}}  ,  \underbrace{q(a,y_{0},z_{0}) }_{\ins_{1}\setminus \dbii} ,  \underbrace{q(a,y_{1},z_{1})}_{\ins_{2}\setminus \ins_{1}} ,  \underbrace{s(b,y_{2},z_{2})}_{\ins_{3}\setminus \ins_{2}}  ,  \underbrace{t(a,y_{0},b,y_{2},o)}_{\ins_{4}\setminus \ins_{3}}   \}
\text{\ \ \ and }$$
$h_1= h_2= \{x {\mapsto} a\}$, $h_3= \{x {\mapsto} b\}$, $h_4= \{x {\mapsto} a, y {\mapsto} y_0, z{\mapsto} z_0, w{\mapsto} b, u{\mapsto} y_2, v {\mapsto} z_2\}$.

\smallskip

\noindent Note that this derivation is not greedy because
\vspace{-2ex}
$$
h_{4}(\fr(\rho_{4})) = h_{4}(\{x,y,w,u\}) = \{a,\overbrace{y_{0}}^{\mathclap{\hspace{16ex}\in\, \nullset(\overline{h}_{1}(\head(\rho_{1})))}},b,\underbrace{y_{2}}_{\mathclap{\hspace{16ex} \in\, \nullset(\overline{h}_{3}(\head(\rho_{2})))}}\}\\ 
$$

\noindent That is to say, the image of the frontier from the last rule application (i.e. the application of $\rho_{4}$) contains nulls introduced by \emph{two} previous rule applications (as opposed to containing nulls from just a single previous rule application), namely, the first application of $\rho_{1}$ and the application of $\rho_{2}$. In contrast, the following is an example of a greedy derivation
 $$
\der_{2} = \dbii, (\rho_{3},h_{1}',\ins_{1}'), (\rho_{1},h_{2}',\ins_{2}'), (\rho_{4},h_{3}',\ins_{3}'),\text{\ \  with }
$$
$$
\ins_3' = \{ \underbrace{p(a),r(b)}_{\smash{\dbii}}  ,  \underbrace{q(a,y_{0},z_{0}),s(b,y_{2},z_{2})}_{\ins_{1}'\setminus \dbii} ,  \underbrace{q(a,y_{1},z_{1})}_{\ins_{2}'\setminus \ins_{1}'} ,  \underbrace{t(a,y_{0},b,y_{2},o)}_{\ins_{3}'\setminus \ins_{2}'}\}
\text{\ \ \ and }$$
$h'_1= \{x {\mapsto} a,y {\mapsto} b\}$, $h'_2= \{x {\mapsto} a\}$, $h'_3= \{x {\mapsto} a, y {\mapsto} y_0, z{\mapsto} z_0, w{\mapsto} b, u{\mapsto} y_2, v {\mapsto} z_2\}$.

\medskip
\noindent Greediness of $\der_{2}$ follows from the frontier of any applied rule being mapped to nothing but constants and/or nulls introduced by a sole previous rule application.

\begin{definition}[(Weakly) Greedy Bounded-Treewidth Set] A rule set $\rs$ is a \emph{greedy bounded-treewidth set ($\gbts$)} \iffi if $\db \drel{\der}{\rs} \ins$, then $\der$ is greedy. $\rs$ is a \emph{weakly greedy bounded-treewidth set ($\wgbts$)} \iffi if $\db \drel{\der}{\rs} \ins$, then there exists some greedy $\rs$-derivation $\der'$ such that $\db \drel{\der'}{\rs} \ins$.
\end{definition}

\begin{remark} Observe that $\gbts$ and $\wgbts$ are characterized on the basis of derivations starting from given \emph{databases} only, that is, derivations of the form $\ins_{0},(\rho_{1},h_{1},\ins_{1}), \ldots, (\rho_{n},h_{n},\ins_{n})$ where $\ins_{0} = \db$ is a database. In such a case, a derivation of the above form is greedy \iffi for each $i$ with $0 < i \leq n$, there exists a $j < i$ such that $h_{i}(\fr(\rho_{i})) \subseteq \nullset(\overline{h}_{j}(\head(\rho_{j}))) \cup \conset(\db,\rs)$ 
as databases only contain constants (and not nulls) by definition.
\end{remark}

As noted above, it is straightforward to show that $\wgbts$ subsumes $\gbts$. 
%
%
 Still, establishing that $\wgbts$ strictly subsumes $\gbts$, i.e. there are rule sets within $\wgbts$ that are outside $\gbts$, requires more effort. As it so happens, the rule set $\rsii$ (defined above) serves as such a rule set, admitting non-greedy $\rsii$-derivations, but where it can be shown that every instance derivable using the rule set admits a greedy $\rsii$-derivation. As a case in point, observe that the $\rsii$-derivations $\der_{1}$ and $\der_{2}$ both derive the same instance $\ins_{4} = \ins_{3}'$, however, $\der_{1}$ is a non-greedy $\rsii$-derivation of the instance and $\der_{2}$ is a greedy $\rsii$-derivation of the instance. Clearly, the existence of the non-greedy $\rsii$-derivation $\der_{1}$ witnesses that $\rsii$ is not $\gbts$. To establish that $\rsii$ still falls within the $\wgbts$ class, we show that every non-greedy $\rsii$-derivation can be transformed into a greedy $\rsii$-derivation using two operations: (i) rule permutations and (ii) rule replacements.
 
 Regarding rule permutations, we consider under what conditions we may swap consecutive applications of rules in a derivation to yield a new derivation of the same instance. For example, in the $\rsii$-derivation $\der_{1}$ above, we may swap the consecutive applications of $\rho_{1}$ and $\rho_{2}$ to obtain the following derivation: 
$$\der_{1}' = \dbii, (\rho_{1},h_{1},\ins_{1}), (\rho_{2},h_{3},\ins_{1} \cup (\ins_{3} \setminus \ins_{2})),(\rho_{1},h_{2},\ins_{3}), (\rho_{4},h_{4},\ins_{4}).$$
 $\ins_{1} \cup (\ins_{3} \setminus \ins_{2}) = \{p(a),r(b),q(a,y_{0},z_{0}),s(b,y_{2},z_{2})\}$ is derived by applying $\rho_{2}$ and the subsequent application of $\rho_{1}$ reclaims the instance $\ins_{3}$. Therefore, the same instance $\ins_{4}$ remains the conclusion. Although one can confirm that $\der_{1}'$ is indeed an \mbox{$\rsii$-derivation}, thus serving as a successful example of a rule permutation (meaning, the rule permutation yields another $\rsii$-derivation), the following question still remains: for a rule set $\rs$, under what conditions will permuting rules within a given $\rs$-derivation always yield another $\rs$-derivation?
 
 We pose an answer to this question, formulated as the \emph{permutation lemma} below, which states that an application of a rule $\rho$ may be permuted before an application of a rule $\rho'$ so long as the former rule does not depend on the latter (in the sense formally defined in \sect~\ref{sec:prelims} based on the work of Baget~\cite{Bag04}). Furthermore, it should be noted that such rule permutations preserve the greediness of derivations. In the context of the above example, $\rho_{2}$ may be permuted before $\rho_{1}$ in $\der_{1}$ because the former does not depend on the latter.

\begin{lemma}[Permutation Lemma]\label{lem:permutation-lemma}
Let $\rs$ be a rule set with $\ins_{0}$ an instance. Suppose we have a (greedy) $\rs$-derivation of the following form:
$$
\ins_{0}, \ldots, (\rho_{i},h_{i},\ins_{i}), (\rho_{i+1},h_{i+1},\ins_{i+1}), \ldots, (\rho_{n},h_{n},\ins_{n})
$$
If $\rho_{i+1}$ does not depend on $\rho_{i}$, then the following is a (greedy) $\rs$-derivation too:
$$
\ins_{0}, \ldots, (\rho_{i+1},h_{i+1},\ins_{i-1} \cup (\ins_{i+1} \setminus \ins_{i})), (\rho_{i},h_{i},\ins_{i+1}), \ldots, (\rho_{n},h_{n},\ins_{n}).
$$
\end{lemma}

 As a consequence of the above lemma, rules may always be permuted in a given $\rs$-derivation so that its structure mirrors the graph of rule dependencies $\grd{\rs}$ (defined in \sect~\ref{sec:prelims}). That is, given a rule set $\rs$ and an $\rs$-derivation $\der$, we may permute all applications of rules  serving as sources in $\grd{\rs}$ (which do not depend on any rules in $\rs$) to the beginning of $\der$, followed by all rule applications that depend only on sources, and so forth, with any applications of rules serving as sinks in $\grd{\rs}$ concluding the derivation. For example, in the graph of rule dependencies of $\rsii$, the rules $\rho_{1}$, $\rho_{2}$, and $\rho_{3}$ serve as source nodes (they do not depend on any rules in $\rsii$) and the rule $\rho_{4}$ is a sink node de\-pen\-ding on each of the aforementioned three rules, i.e. $\grd{\rsii} = (\grdv,\grde)$ with $\grdv = \{\rho_{1}, \rho_{2}, \rho_{3}, \rho_{4}\}$ and $\grde = \{(\rho_{i},\rho_{4}) \ | \ 1 \leq i \leq 3\}$. Hence, in any given $\rsii$-derivation $\der$, any application of $\rho_{1}$, $\rho_{2}$, or $\rho_{3}$ can be permuted backward (toward the beginning of $\der$) and any application of $\rho_{4}$ can be permuted forward (toward the end of $\der$).

 Beyond the use of rule permutations, we also transform $\rsii$-derivations by making use of rule replacements. In particular, observe that $\head(\rho_{3})$ and $\body(\rho_{3})$ correspond to conjunctions of $\head(\rho_{1})$ and $\head(\rho_{2})$, and $\body(\rho_{1})$ and $\body(\rho_{2})$, respectively. Thus, we can replace the first application of $\rho_{1}$ and the succeeding application of $\rho_{2}$ in $\der_{1}'$ above by a single application of $\rho_{3}$, thus yielding the $\rsii$-derivation
$
\der_{1}'' = \dbii, (\rho_{3},h,\ins_{1} \cup (\ins_{3} \setminus \ins_{2})), (\rho_{1},h_{2},\ins_{3}), (\rho_{4},h_{4},\ins_{4}),
$
 where $h(x) = a$ and $h(y) = b$. Interestingly, inspecting the above \mbox{$\rsii$-derivation}, one will find that it is identical to the greedy \mbox{$\rsii$-derivation} $\der_{2}$ defined earlier in the section, and so, we have shown how to take a non-greedy \mbox{$\rsii$-derivation} (viz. $\der_{1}$) and transform it into a greedy \mbox{$\rsii$-derivation} (viz. $\der_{2}$) by means of rule permutations and replacements. In the same way, one can prove in general that any non-greedy \mbox{$\rsii$-derivation} can be transformed into a greedy \mbox{$\rsii$-derivation}, thus giving rise to the following theorem, and demonstrating that $\rsii$ is indeed $\wgbts$. 

\begin{theorem}\label{lem:wgbts-contains-gbts} $\rsii$ is $\wgbts$, but not $\gbts$. Thus, $\wgbts$ properly subsumes $\gbts$.
\end{theorem}

\section{Derivation Graphs}\label{sec:derivation-graphs}

 We now discuss \emph{derivation graphs} -- a concept introduced by Baget et al.~\cite{BagLecMugSal11} and used to establish that certain classes of rule sets (e.g. weakly frontier guarded rule sets~\cite{CaliGK13}) exhibit universal models of bounded treewidth. A derivation graph has the structure of a directed acyclic graph and encodes \emph{how} atoms are derived throughout the course of an $\rs$-derivation. By applying so-called \emph{reduction operations}, a derivation graph may (under certain conditions) be transformed into a treelike graph that serves as a tree decomposition of an $\rs$-derivable instance. 
 
 Below, we define derivation graphs and discuss how such graphs are transformed into tree decompositions by means of reduction operations. To increase comprehensibility, we provide an example of a derivation graph (shown in \Cref{fig:der-graph-example}) and give an example of applying each reduction operation (shown in \Cref{fig:reducing-der-graph}). After, we identify two (query-decidable) classes of rule sets on the basis of derivation graphs, namely, \emph{cycle-free derivation graph sets} ($\adgs$) and \emph{weakly cycle-free derivation graph sets} ($\wcdgs$). Despite their prima facie distinctness, the $\adgs$ and $\wcdgs$ classes coincide with $\gbts$ and $\wgbts$ classes, respectively, thus showing how the latter classes can be characterized in terms of derivation graphs. Let us now formally define derivation graphs, and after, we will demonstrate the concept by means of an example.

\begin{definition}[Derivation Graph]\label{def:derivation-graph} Let $\db$ be a database, $\rs$ be a rule set, $\cons = \conset(\db,\rs)$, and $\der$ be some $\rs$-derivation $\db, (\rho_{1}, h_{1}, \ins_{1}), \ldots, (\rho_{n}, h_{n}, \ins_{n})$. The \emph{derivation graph of $\der$} is the tuple $\dg := (\nodes,\edges,\at,\lbl)$, where $\nodes := \{X_{0}, \ldots,X_{n}\}$ is a finite set of \emph{nodes}, $\edges \subseteq \nodes \times \nodes$ is a set of \emph{arcs}, and the functions $\at: \nodes \to 2^{\ins_{n}}$ and $\lbl: \edges \to 2^{\termset(\ins_{n})}$ decorate nodes and arcs, respectively, such that:
\begin{enumerate}

\item $\at(X_{0}) := \db$ and $\at(X_{i}) = \ins_{i} \setminus \ins_{i-1}$;

\item $(X_{i},X_{j}) \in \edges$ \iffi there is a $p(\vec{t}) \in \at(X_{i})$ and a frontier atom $p(\vec{t'})$ in $\rho_{j}$ such that $h_{j}(p(\vec{t'})) = p(\vec{t})$. We then set $\lbl(X_{i},X_{j})=
\Big(h_{j}\big(\varset(p(\vec{t'})) \cap \fr(\rho_{j})\big)\Big) \setminus \cons$.

\end{enumerate}
We refer to $X_{0}$ as the \emph{initial node} and define the set of \emph{non-constant terms} associated with a node to be $\nonc(X) = \termset(X) \setminus C$ where $\termset(X_{i}) := \termset(\at(X_{i})) \cup \cons$.
\end{definition}
Toward an example, assume $\dbiii = \{p(a,b)\}$ and $\rsiii = \{\rho_{1},\rho_{2},\rho_{3},\rho_{4}\}$ where 
$$
\begin{array}{rl@{\qquad}rl}
\rho_{1} = & \, p(x,y) \rightarrow \exists z . q(y,z)     &    \rho_{3} = & \, r(x,y) \land q(z,x)  \rightarrow s(x,y) \\[0.0ex]
\rho_{2} = & \, q(x,y)  \rightarrow \exists z .  (r(x,y) \land r(y,z))     &    \rho_{4} = & \, r(x,y) \land s(z,w)  \rightarrow t(y,w) \\
\end{array}
$$
Let us consider the following derivation:
$$
\der = \dbiii, (\rho_{1},h_{1},\ins_{1}), (\rho_{2},h_{2},\ins_{2}), (\rho_{3},h_{3},\ins_{3}), (\rho_{4},h_{4},\ins_{4}) \text{\ \ with}
$$
$$
\ins_4 = \{ \underbrace{p(a,b)}_{\smash{\dbiii}}  ,  \underbrace{q(b,z_{0}) }_{\ins_{1}\setminus \dbiii} ,  \underbrace{r(b,z_{0}),r(z_0,z_{1})}_{\ins_{2}\setminus \ins_{1}} ,  \underbrace{s(z_{0},z_{1})}_{\ins_{3}\setminus \ins_{2}}  ,  \underbrace{t(z_{0},z_{1})}_{\ins_{4}\setminus \ins_{3}}   \}
\text{\ \ \ and }$$ 
$h_1= \{x {\mapsto} a,y {\mapsto} b\}$, $h_2= \{x {\mapsto} b,y {\mapsto} z_0\}$, $h_3= \{x {\mapsto} z_0, y {\mapsto} z_1, z{\mapsto} b\}$, as well as $h_4= \{x {\mapsto} b, y {\mapsto} z_0, z{\mapsto} z_0, w{\mapsto} z_1\}$.
 The derivation graph $\dg = (\nodes,\edges,\at,\lbl)$ corresponding to $\der$ is shown in \Cref{fig:der-graph-example} and has fives nodes, $\nodes = \{X_{0},X_{1},X_{2},X_{3},X_{4}\}$. Each node $X_{i} \in \nodes$ is associated with a set $\at(X_{i})$ of atoms depicted in the associated circle (e.g. $\at(X_{2}) = \{r(b,z_{0}),r(z_{0},z_{1})\}$), and each arc $(X_{i},X_{j}) \in \edges$ is represented as a directed arrow with $\lbl(X_{i},X_{j})$ shown as the associated set of terms (e.g. $\lbl(X_{3},X_{4}) = \{z_{1}\}$). For each node $X_{i} \in \nodes$, the set $\termset(X_{i})$ of terms associated with the node is equal to $\termset(\at(X_{i})) \cup \{a,b\}$ (e.g. $\termset(X_{3}) = \{z_{0},z_{1},a,b\}$) since $C = \conset(\dbiii,\rsiii) = \{a,b\}$.\\[-5ex]

\begin{figure}[h]
\centerline{
\scalebox{0.6}{
\large\begin{tikzpicture}[
dot/.style = {circle, minimum size=#1,
              inner sep=0pt, outer sep=0pt},
dot/.default = 1.5cm  
                    ] 

\node[dot,draw=black,label=above left:$X_{0}$] (z0) [] {$p(a,b)$};

\node[dot,draw=black,label=below left:$X_{1}$] (z1) [below of=z0, yshift=-1.5cm] {$q(b,z_{0})$};
\draw[->] (z0) -- (z1) node [midway,xshift=-1em] {$\emptyset$};

\node[dot,draw=black,label=above right:$X_{2}$] (z2) [right of=z0, xshift=1.5cm] 
{$\myoverset{r(b,z_{0})}{r(z_{0},z_{1})}$};
\draw[->] (z1) -- (z2) node [midway,yshift=1em,xshift=-.5em] {$\{z_{0}\}$};

\node[dot,draw=black,label=below right:$X_{3}$] (z3) [right of=z1, xshift=1.5cm]
 {$s(z_{0},z_{1})$};
\draw[->] (z1) -- (z3) node [midway,yshift=1em] {$\{z_{0}\}$};
\draw[->] (z2) -- (z3) node [midway,xshift=2em] {$\{z_{0},z_{1}\}$};

\node[dot,draw=black,label=right:$X_{4}$] (z4) [right of=z2, xshift=1.5cm, yshift=-1.15cm]
 {$t(z_{0},z_{1})$};
\draw[->] (z2) -- (z4) node [midway,yshift=1em] {$\{z_{0}\}$};
\draw[->] (z3) -- (z4) node [midway,yshift=-1em] {$\{z_{1}\}$};

\end{tikzpicture}}
}
\caption{The derivation graph $\dg$.}
\label{fig:der-graph-example}
\vspace{-2ex}
\end{figure}

 As can be witnessed via the above example, derivation graphs satisfy a set of properties akin to those characterizing tree decompositions~\cite[\prp~12]{BagLecMugSal11}. 

\begin{lemma}[Decomposition Properties]\label{lem:decomposition-properties}
Let $\db$ be a database, $\rs$ be a rule set, and $C = \conset(\db,\rs)$. If $\db \drel{\der}{\rs} \ins$, then $\dg$ satisfies the following properties:
\begin{enumerate}

\item $\bigcup_{X_{n} \in \nodes} \termset(X_{n}) = \termset(\ins)$;

\item For each $p(\vec{t}) \in \ins$, there is an $X_{n} \in \nodes$ such that $p(\vec{t}) \in \at(X_{n})$;

\item For each term $x \in \nonc(\ins)$, the subgraph of $\dg$ induced by the nodes $X_{n}$ such that $x \in \nonc(X_{n})$ is connected;

\item For each $X_{n} \in \nodes$ the size of $\termset(X_{n})$ is bounded by an integer that only depends on the size of $(\db,\rs)$, \textit{viz.} $\max\{|\termset(\db)|,|\termset(\head(\rho_{i}))|_{\rho_{i} \in \rs}\} + |\cons|$.

\end{enumerate}
\end{lemma}


 Let us now introduce our set of \emph{reduction operations}. As remarked above, in certain circumstances such operations can be used to transform derivation graphs into tree decompositions of an instance.

We make use of three reduction operations, namely, (i) \emph{arc removal}, denoted $\AR^{[i,j]}$, (ii) \emph{term removal}, denoted $\TR^{[i,j,k,t]}$, and (iii) \emph{cycle removal}, denoted $\CR^{[i,j,k,\ell]}$. The first two reduction operations were already proposed by Baget et al.~\cite{BagLecMugSal11} (they presented $\TR$ and $\AR$ as a single operation called \emph{redundant arc removal}), whereas cycle removal is introduced by us as a new operation that will assist us in characterizing $\gbts$ and $\wgbts$ in terms of derivation graphs.\footnote{Beyond $\TR$ and $\AR$, we note that Baget et al.~\cite{BagLecMugSal11} introduced an additional reduction operation, referred to as \emph{arc contraction}. We do not consider this rule here however as it is unnecessary to characterize $\gbts$ and $\wgbts$ in terms of derivation graphs and prima facie obstructs the proof of \thm~\ref{thm:(w)gbts-(w)adgs-equivalence}.}

\begin{definition}[Reduction Operations]\label{def:reduction-operations} Let $\db$ be a database, $\rs$ be a rule set, $\db \drel{\der}{\rs} \ins_{n}$, and $\dg$ be the derivation graph of $\der$. We define the set $\mathsf{RO}$ of \emph{reduction operations} as
$
\{\AR^{[i,j]}\!, \TR^{[i,j,k,t]}\!, \CR^{[i,j,k,\ell]} \,|\, i,j,k,\ell {\,\leq\,} n,\, t{\,\in\,}  \termset(\ins_{n})\},
$
 whose effect is further specified below. We let $\ro\os(\dg)$ denote the output of applying the operation $\ro$ to the (potentially reduced) derivation graph $\os(\dg) = (\nodes,\edges,\at,\lbl)$, where $\os \in \mathsf{RO}^{*}$ is a \emph{reduction sequence}, that is, $\os$ is a (potentially empty) sequence of reduction operations.
 
\begin{enumerate}

\item Arc Removal $\AR^{[i,j]}$: Whenever $(X_{i},X_{j}) \in \edges$ and $\lbl(X_{i},X_{j}) = \emptyset$, then $\AR^{[i,j]}\os(\dg) := (\nodes,\edges',\at,\lbl')$ where $\edges' := \edges \setminus \{(X_{i},X_{j})\}$ and $\lbl' = \lbl \restriction \edges'$.

\item Term Removal $\TR^{[i,j,k,t]}$: If $(X_{i},X_{k}),(X_{j},X_{k}) \in \edges$ with $X_{i} \neq X_{j}$ and  $t \in \lbl(X_{i},X_{k}) \cap \lbl(X_{j},X_{k})$, then $\TR^{[i,j,k,t]}\os(\dg) := (\nodes,\edges,\at,\lbl')$ where $\lbl'$ is obtained from $\lbl$ by removing $t$ from $\lbl(X_{j},X_{k})$.

\item Cycle Removal $\CR^{[i,j,k,\ell]}$: If $(X_{i},X_{k}),(X_{j},X_{k}) \in \edges$ and there exists a node $X_{\ell} \in \nodes$ with $\ell \ilt k$ such that 
$
\lbl(X_{i},X_{k}) \cup \lbl(X_{j},X_{k}) \subseteq \termset(X_{\ell})
$
 then, $\CR^{[i,j,k,\ell]}\os(\dg) := (\nodes,\edges',\at,\lbl')$ where
$$
\edges' := \big(\edges \setminus \{(X_{i},X_{k}),(X_{j},X_{k}) \}\big) \cup \{(X_{\ell},X_{k})\}
$$
 and $\lbl'$ is obtained from $\lbl \restriction \edges'$ by setting $L(X_{\ell},X_{k})$ to $\lbl(X_{i},X_{k}) \cup \lbl(X_{j},X_{k})$.

\end{enumerate}
Last, we say that a reduction sequence $\os \in \mathsf{RO}^{*}$ is a \emph{complete reduction sequence} relative to a derivation graph $\dg$ \iffi $\os(\dg)$ is cycle-free. 
\end{definition}

\begin{remark} When there is no danger of confusion, we will take the liberty to write $\TR$, $\AR$, and $\CR$ without superscript parameters. That is, given a derivation graph $\dg$, the (reduced) derivation graph $\CR\TR(\dg)$ is obtained by applying an instance of $\TR$ followed by an instance of $\CR$ to $\dg$. When applying a reduction operation we always explain \emph{how} it is applied, so the exact operation is known.
\end{remark}

\begin{figure}[t]
\resizebox{\columnwidth}{!}{
\includegraphics{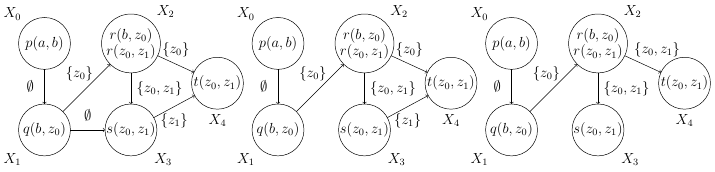}
}

\caption{Left to right: reduced derivation graphs $\TR\!(\dg\!)$, $\AR\!\TR\!(\dg\!)$, and $\CR\!\AR\!\TR\!(\dg\!)$.}
\label{fig:reducing-der-graph}
\end{figure}

 We now describe the functionality of each reduction operation and illustrate each by means of an example. We will apply each to transform the derivation graph $\dg$ (shown in \fig~\ref{fig:der-graph-example}) into a tree decomposition of $\ins_{4}$ (which was defined above). 
 The $\TR$ operation deletes a term $t$ within the intersection of the sets labeling two converging arcs. For example, we may apply $\TR$ to the derivation graph $\dg$ from \fig~\ref{fig:der-graph-example}, deleting the term $z_{0}$ from the label of the arc $(X_{1},X_{3})$, and yielding the reduced derivation graph $\TR(\dg)$, which is shown first in \fig~\ref{fig:reducing-der-graph}. We may then apply $\AR$ to $\TR(\dg)$, deleting the arc $(X_{1},X_{3})$, which is labeled with the empty set, to obtain the reduced derivation graph $\AR\TR(\dg)$ shown middle in \fig~\ref{fig:reducing-der-graph}. 
 
 The $\CR$ operation is more complex and works by considering two converging arcs $(X_{i},X_{k})$ and $(X_{j},X_{k})$ in a (reduced) derivation graph. If there exists a node $X_{\ell}$ whose index $\ell$ is less than the index $k$ of the child node $X_{k}$ and $\lbl(X_{i},X_{k}) \cup \lbl(X_{j},X_{k}) \subseteq \termset(X_{\ell})$, then the converging arcs $(X_{i},X_{k})$ and $(X_{j},X_{k})$ may be deleted and the arc $(X_{\ell},X_{k})$ introduced and labeled with $\lbl(X_{i},X_{k}) \cup \lbl(X_{j},X_{k})$. As an example, the reduced derivation graph $\CR\AR\TR(\dg)$ (shown third in \fig~\ref{fig:reducing-der-graph}) is obtained from $\AR\TR(\dg)$ (shown middle in \fig~\ref{fig:reducing-der-graph}) by applying $\CR$ in the following manner to the convergent arcs $(X_{2},X_{4})$ and $(X_{3},X_{4})$: since for $X_{2}$ (whose index $2$ is less than the index $4$ of $X_{4}$) $\lbl(X_{2},X_{4}) \cup \lbl(X_{3},X_{4}) \subseteq \termset(X_{2})$, we may delete the arcs $(X_{2},X_{4})$ and $(X_{3},X_{4})$ and introduce the arc $(X_{2},X_{4})$ labeled with $\lbl(X_{2},X_{4}) \cup \lbl(X_{3},X_{4}) = \{z_{0}\} \cup \{z_{1}\} = \{z_{0},z_{1}\}$. Observe that the reduced derivation graph $\CR\AR\TR(\dg)$ is free of cycles, witnessing that $\os = \CR\AR\TR$ is a complete reduction sequence relative to $\dg$. Moreover, if we replace each node by the set of its terms and disregard the labels on arcs, then $\os(\dg)$ can be read as a tree decomposition of $\ins_{4}$. In fact, one can show that every reduced derivation graph satisfies the decomposition properties mentioned in \lem~\ref{lem:decomposition-properties} above.

\begin{lemma}\label{lem:reduct-ops-preserve-decomp-props}
Let $\db$ be a database and $\rs$ be a rule set. If $\db \drel{\rs}{\der} \ins$, then for any reduction sequence $\os$, $\os(\dg) = (\nodes,\edges,\at,\lbl)$ satisfies the decomposition properties 1-4 in \lem~\ref{lem:decomposition-properties}.
\end{lemma}

 As illustrated above, derivation graphs can be used to derive tree decompositions of $\rs$-derivable instances. By the fourth decomposition property (see \lem~\ref{lem:decomposition-properties} above), the width of such a tree decomposition is bounded by a constant that depends only on the given knowledge base. Thus, if a rule set $\rs$ always yields derivation graphs that are reducible to \emph{cycle-free} graphs -- meaning that (un)directed cycles do not occur within the graph -- then all $\rs$-derivable instances have tree decompositions that are uniformly bounded by a constant. 
 This establishes that the rule set $\rs$ falls within the $\bts$ class, confirming that query entailment is decidable with $\rs$. We define two classes of rule sets by means of reducible derivation graphs:

\begin{definition}[(Weakly) Cycle-free Derivation Graph Set]\label{def:adgs-wadgs} A rule set $\rs$ is a \emph{cycle-free derivation graph set ($\cdgs$)} \iffi if $\db\! \drel{\der}{\rs\!} \ins$, then $\dg$ can be re\-duced to a cycle-free graph by the reduction operations. $\rs$ is a \emph{weakly cycle-free de\-rivation graph set ($\wcdgs$)} \iffi if $\db\! \drel{\der}{\rs\!} \ins$, then there is a derivation $\der'\!$ where $\db\! \drel{\der'}{\rs\!} \ins$ and $G_{\!\der'}\!$ can be reduced to a cycle-free graph by the reduction operations.
\end{definition}

 It is straightforward to confirm that $\wcdgs$ subsumes $\cdgs$, and that both classes are subsumed by $\bts$.

\begin{proposition}\label{prop:(w)adgs-relationships}
Every $\cdgs$ rule set is $\wcdgs$ and every $\wcdgs$ rule set is $\bts$. 
\end{proposition}

 Furthermore, as mentioned above, $\gbts$ and $\wgbts$ coincide with $\cdgs$ and $\wcdgs$, respectively. By making use of the $\CR$ operation, one can show that the derivation graph of any greedy derivation is reducible to a cycle-free graph, thus establishing that $\gbts \subseteq \cdgs$ and $\wgbts \subseteq \wcdgs$. To show the converse (i.e. that $\cdgs \subseteq \gbts$ and $\wcdgs \subseteq \wgbts$) however, requires more work. In essence, one shows that for every (non-source) node $X_{i}$ in a cycle-free (reduced) derivation graph there exists another node $X_{j}$ such that $j < i$ and the frontier of the atoms in $\at(X_{i})$ only consist of constants and/or nulls introduced by the atoms in $\at(X_{j})$. This property is preserved under \emph{reverse} applications of the reduction operations, and thus, one can show that if a derivation graph is reducible to a cycle-free graph, then the above property holds for the original derivation graph, implying that the derivation graph encodes a greedy derivation. Based on such arguments, one can prove the following:

\begin{theorem}\label{thm:(w)gbts-(w)adgs-equivalence}
$\gbts$ coincides with $\cdgs$ and $\wgbts$ coincides with $\wcdgs$. Membership in $\cdgs$, $\gbts$, $\wcdgs$, or $\wgbts$ warrants decidable BCQ entailment.
\end{theorem}
Note that by \Cref{lem:wgbts-contains-gbts}, this also implies that $\wcdgs$ properly contains $\cdgs$.

 An interesting consequence of the above theorem concerns the redundancy of $\AR$ and $\TR$ in the presence of $\CR$. In particular, since we know that (i) if a derivation graph can be reduced to a cycle-free graph, then the derivation graph encodes a greedy derivation, and (ii) the derivation graph of any greedy derivation can be reduced to an cycle-free graph by means of applying the $\CR$ operation only, it follows that if a derivation graph can be reduced to a cycle-free graph, then it can be reduced by only applying the $\CR$ operation. We refer to this phenomenon as \emph{reduction-admissibility}, which is defined below. 

\begin{definition}[Reduction-admissible] Suppose $S_{1} = \{\rop{i} \ | \ 1 \leq i \leq n\}$ and $S_{2} = \{\rop{j} \ | \ n+1 \leq j \leq k\}$ are two sets of reduction operations. We say that $S_{1}$ is \emph{reduction-admissible} relative to $S_{2}$ \iffi for any rule set $\rs$ and $\rs$-derivation $\der$, if $\dg$ is reducible to a cycle-free graph with $S_{1} \cup S_{2}$, then $\dg$ is reducible to a cycle-free graph with just $S_{2}$.
\end{definition}

\begin{corollary}\label{corollaries}
$\{\TR,\AR\}$ is reduction-admissible relative to $\CR$.
\end{corollary}

\section{Conclusion}\label{sec:conclusion}

In this paper, we revisited the concept of a \emph{greedy} derivation, which immediately gives rise to a bounded-width tree decomposition of the constructed instance. This well-established notion allows us to categorize rule sets as being \emph{(weakly) greedy bounded treewidth sets} ($\wandgbts$), if all (some) derivations of a derivable instance are guaranteed to be greedy, irrespective of the underlying database. By virtue of being subsumed by $\bts$, these classes warrant decidability of BCQ entailment, while at the same time subsuming various popular rule languages, in particular from the guarded family. 

By means of an example together with a proof-theoretic argument, we exposed that $\wgbts$ strictly generalizes $\gbts$. In pursuit of a better understanding and more workable methods to detect and analyze $\wandgbts$ rule sets, we resorted to the previously proposed notion of \emph{derivation graphs}. Through a refinement of the set of reduction methods for derivation graphs, we were able to make more advanced use of this tool, leading to the definition of \emph{(weakly) cycle-free derivation graph sets} ($\wandcdgs$) of rules, of which we were then able to show the respective coincidences with $\wandgbts$. This way, we were able to establish alternative characterizations of $\gbts$ and $\wgbts$ by means of derivation graphs. En passant, we found that the newly introduced \emph{cycle removal} reduction operation over derivation graphs is sufficient by itself and makes the other operations redundant.

For future work, we plan to put our newly found characterizations to use. In particular, we aim to investigate if a rule set's membership in $\gbts$ or $\wgbts$ is decidable. For $\gbts$, this has been widely conjectured, but never formally established. In the positive case, derivation graphs might also be leveraged to pinpoint the precise complexity of the membership problem. We are also confident that the tools and insights in this paper -- partially revived, partially upgraded, partially newly developed -- will prove useful in the greater area of static analysis of existential rule sets. On a general note, we feel that the field of proof theory has a lot to offer for knowledge representation, whereas the cross-fertilization between these disciplines still appears to be underdeveloped.  







\bibliographystyle{splncs04}
\bibliography{bibliography}

\appendix



\section{Proofs for \Cref{sec:greedy}}

\begin{customlem}{\ref{lem:permutation-lemma}}[Permutation Lemma] Let $\rs$ be a rule set with $\ins_{0}$ an instance. Suppose we have a (greedy) $\rs$-derivation of the following form:
$$
\ins_{0}, \ldots, (\rho_{i},h_{i},\ins_{i}), (\rho_{i+1},h_{i+1},\ins_{i+1}), \ldots, (\rho_{n},h_{n},\ins_{n})
$$
If $\rho_{i+1}$ does not depend on $\rho_{i}$, then the following is a (greedy) $\rs$-derivation too:
$$
\ins_{0}, \ldots, (\rho_{i+1},h_{i+1},\ins_{i-1} \cup (\ins_{i+1} \setminus \ins_{i})), (\rho_{i},h_{i},\ins_{i+1}), \ldots, (\rho_{n},h_{n},\ins_{n}).
$$
\end{customlem}

\begin{proof} By assumption, $\rho_{i+1}$ does not depend on $\rho_{i}$, implying $h_{i+1}(\body(\rho_{i+1})) \subseteq \ins_{i-1}$. Hence, we may apply $\rho_{i+1}$ with $h_{i+1}$ directly to $\ins_{i-1}$ yielding the instance $\ins_{i}' = \ins_{i-1} \cup (\ins_{i+1} \setminus \ins_{i})$. Since $h_{i}(\body(\rho_{i})) \subseteq \ins_{i-1} \subseteq \ins_{i}'$, we may apply $\rho_{i}$ directly after $\rho_{i+1}$ yielding the instance $\ins_{i+1}$. Moreover, if $\der$ is greedy, then (i) $h_{i+1}(\fr(\rho_{i+1})) \subseteq \nullset(\overline{h}_{j}(\head(\rho_{j}))) \cup \conset(\ins_{0},\rs) \cup \nullset(\ins_{0})$ for some $j < i+1$, and (ii) $h_{i}(\fr(\rho_{i})) \subseteq \nullset(\overline{h}_{k}(\head(\rho_{k}))) \cup \conset(\ins_{0},\rs) \cup \nullset(\ins_{0})$ for some $k < i$. As $\rho_{i+1}$ does not depend on $\rho_{i}$, it must be the case that $j \neq i$, and so, we have that $\der'$ will be greedy as well since (i) and (ii) will hold for $j,k < i$ in $\der'$. 
\end{proof}

\begin{lemma}\label{prop:gbts-is-wgbts}
Let $\rs$ be a rule set. If $\rs$ is $\gbts$, then $\rs$ is $\wgbts$.
\end{lemma}

\begin{proof} Let $\db$ be a database and $\rs$ be a $\gbts$ rule set. If $\db \drel{\der}{\rs} \ins$, then $\der$ is greedy as $\rs$ is $\gbts$. Hence, there exists a greedy $\rs$-derivation (viz. $\der$) of $\ins$ from $\db$, showing that $\rs$ is $\wgbts$ as well.
\end{proof}

\begin{customthm}{\ref{lem:wgbts-contains-gbts}} $\rsii$ is $\wgbts$, but not $\gbts$. Thus, $\wgbts$ properly subsumes $\gbts$.
\end{customthm}

\begin{proof} We know that $\wgbts$ subsumes $\gbts$ by \lem~\ref{prop:gbts-is-wgbts} above, however, to show that $\wgbts$ properly subsumes $\gbts$, we prove that $\rsii$ is $\wgbts$, but not $\gbts$. Therefore, let $\db$ be an arbitrary database and $\ins$ be an instance such that there exists an $\rsii$-derivation $\der_{0}$ of $\ins$ from $\db$. We show by induction on the length of $\der_{0}$ that a greedy $\rsii$-derivation of $\ins$ from $\db$ can always be found.

\textit{Base case.} Any $\rsii$-derivation of an instance $\ins$ from $\db$ of length $n = 0$ or $n = 1$ is trivially greedy by \dfn~\ref{def:greedy-derivation}.

\textit{Inductive step.} Suppose our derivation $\der_{0}$ is of length $n+1$, that is
$$
\der_{0} = \db, (\rho_{1},h_{1},\ins_{1}), \ldots, (\rho_{n},h_{n},\ins_{n}), (\rho_{n+1},h_{n+1},\ins_{n+1})
$$
By IH, we have that a greedy $\rsii$-derivation $\der_{1}$ of $\ins_{n}$ exists; hence, let $\der_{2} = \der_{1}, (\rho_{n+1},h_{n+1},\ins_{n+1})$ and observe that $\der_{2}$ is a valid $\rsii$-derivation as we already know by the structure of $\der_{0}$ above that $\rho_{n+1}$ is triggered in $\ins_{n}$ with the homomorphism $h_{n+1}$. If the last rule $\rho_{n+1}$ applied in $\der_{2}$ is $\rho_{1}$, $\rho_{2}$, or $\rho_{3}$, then since no such rule depends on any rule in $\rsii$, it must be the case $h_{n+1}(\head(\rho_{n+1})) \subseteq \db$, showing that $\der_{2}$ is greedy. Therefore, let us assume that the last rule $\rho_{n+1}$ applied is $\rho_{4}$. Recall that $\body(\rho_{4}) = \{q(x,y,z), s(w,u,v)\}$, and observe that if $\rho_{4}$ is applied, then $h_{n+1}(q(x,y,z)),h_{n+1}(s(w,u,v)) \in \ins_{n}$. We make a case distinction depending on the how $h_{n+1}(q(x,y,z))$ and $h_{n+1}(s(w,u,v))$ entered into the derivation $\der_{1}$ below:

\begin{enumerate}

\item Suppose that $h_{n+1}(q(x,y,z)),h_{n+1}(s(w,u,v)) \in \db$. Then, $\der_{2}$ is greedy since
\begin{align*}
h_{n+1}(\fr(\rho_{n+1})) & \subseteq \conset(\db)\\
& \subseteq \nullset(\overline{h}_{j}(\head(\rho_{j}))) \cup \conset(\db,\rs)
\end{align*}
for any $j < n+1$.

\item Suppose that $h_{n+1}(q(x,y,z)) \in \db$ and $h_{n+1}(s(w,u,v))$ was introduced by an application of $\rho_{2}$ or $\rho_{3}$ at $j < n+1$ (i.e. $\rho_{j} \in \{\rho_{2},\rho_{3}\}$). Then, $\der_{2}$ is greedy since
\begin{align*}
h_{n+1}(\fr(\rho_{n+1})) & \subseteq \nullset(\overline{h}_{j}(\head(\rho_{j}))) \cup \conset(\db)\\
& \subseteq \nullset(\overline{h}_{j}(\head(\rho_{j}))) \cup \conset(\db,\rs).
\end{align*}

\item Suppose that $h_{n+1}(q(x,y,z))$ was introduced by an application of $\rho_{1}$ or $\rho_{3}$ at $j < n+1$ (i.e. $\rho_{j} \in \{\rho_{1},\rho_{3}\}$) and $h_{n+1}(s(w,u,v)) \in \db$. Then, $\der_{2}$ is greedy since
\begin{align*}
h_{n+1}(\fr(\rho_{n+1})) & \subseteq \nullset(\overline{h}_{j}(\head(\rho_{j}))) \cup \conset(\db)\\
& \subseteq \nullset(\overline{h}_{j}(\head(\rho_{j}))) \cup \conset(\db,\rs).
\end{align*}

\item Suppose that $h_{n+1}(q(x,y,z))$ and $h_{n+1}(s(w,u,v))$ were introduced by a single application of $\rho_{3}$ at $j < n+1$ (i.e. $\rho_{j}  = \rho_{3}$). Then, $\der_{2}$ is greedy since
\begin{align*}
h_{n+1}(\fr(\rho_{n+1})) & \subseteq \nullset(\overline{h}_{j}(\head(\rho_{3}))) \cup \conset(\db)\\
& \subseteq \nullset(\overline{h}_{j}(\head(\rho_{j}))) \cup \conset(\db,\rs).
\end{align*}

\item Suppose that $h_{n+1}(q(x,y,z))$ was introduced by an application of $\rho_{j} \in \{\rho_{1},\rho_{3}\}$ and $h_{n+1}(s(w,u,v))$ was introduced by an application of $\rho_{k} \in \{\rho_{2},\rho_{3}\}$ with $j,k < n+1$. We assume that if $\rho_{3}$ introduced both $h_{n+1}(q(x,y,z))$ and $h_{n+1}(s(w,u,v))$, then both applications of $\rho_{3}$ are distinct, and we assume w.l.o.g. that $j < k$. Since $\rho_{k}$ only depends on the database $\db$, we may repeatedly apply the permutation lemma (\lem~\ref{lem:permutation-lemma}) to $\der_{2}$, permuting the application of $\rho_{k}$ earlier in the derivation until we reach the application of $\rho_{j}$, yielding:
$$
\der_{3} = \db, \ldots, (\rho_{j},h_{j},\ins_{j}), (\rho_{k},h_{k},\ins_{j+1}'), \ldots, (\rho_{n+1},h_{n+1},\ins_{n+1})
$$
where $\ins_{j+1}' = \ins_{j} \cup (\ins_{k} \setminus \ins_{k-1})$. By the permutation lemma, we know that the portion of $\der_{3}'$ up to and including $(\rho_{n},h_{n},\ins_{n})$ is greedy. We have four cases to consider, and in each case, we show how to transform $\der_{3}$ into a greedy derivation $\der_{3}'$ of the same conclusion.

\begin{enumerate}

\item If $\rho_{j} = \rho_{1}$ and $\rho_{k} = \rho_{2}$, then replace $(\rho_{j},h_{j},\ins_{j}), (\rho_{k},h_{k},\ins_{j+1}')$ in $\der_{3}$ with $(\rho_{3},h',\ins_{j+1}')$ where $h'(p(x)) = h_{j}(\body(\rho_{j}))$, $h'(r(y)) = h_{j}(\body(\rho_{k}))$, and $\overline{h'}(\head(\rho_{3})) = \overline{h}_{j}(q(x,y,z)) \land \overline{h}_{k}(s(x,y,z))$. This gives the derivation:
$$
\der_{3}' = \db, \ldots, (\rho_{3},h',\ins_{j+1}'), \ldots, (\rho_{n+1},h_{n+1},\ins_{n+1})
$$
One can confirm that $\der_{3}'$ is indeed a valid derivation as $h'(\body(\rho_{3})) \in \db$, showing that $\rho_{3}$ may be applied where it is. Also,
$$
\ins_{j+1}' =  \ins_{j-1} \cup \{\overline{h}_{j}(q(x,y,z)),\overline{h}_{k}(s(x,y,z))\} = \ins_{j-1} \cup \{\overline{h'}(\head(\rho_{3}))\},
$$
showing that $\ins_{j+1}'$ is indeed derived by applying $\rho_{3}$, and for any application of a rule $\rho_{m}$ with $j < m$ (i.e. for any application of a rule occurring after the application of $\rho_{3}$ displayed in $\der_{3}'$ above) if it previously depended on $\rho_{j}$ or $\rho_{k}$, it will now depend on the above application of $\rho_{3}$, which introduces the same atoms as $\rho_{j}$ and $\rho_{k}$. This also shows that  the portion of $\der_{3}'$ up to and including $(\rho_{n},h_{n},\ins_{n})$ is greedy. Last, it follows that $\rho_{n+1} = \rho_{4}$ now depends on the above application of $\rho_{3}$, showing that
\begin{align*}
h_{n+1}(\fr(\rho_{n+1})) & \subseteq \nullset(\overline{h'}(\head(\rho_{3}))) \cup \conset(\db,\rs),
\end{align*}
and hence, $\der_{3}'$ is greedy.

\item If $\rho_{j} = \rho_{1}$ and $\rho_{k} = \rho_{3}$, then replace $(\rho_{j},h_{j},\ins_{j})$ in $\der_{3}$ with $(\rho_{3},h',\ins_{j}')$ such that $h'(p(x)) = h_{j}(p(x))$, $h'(r(y)) = h_{k}(r(x))$, $\overline{h'}(\head(\rho_{3})) = \overline{h}_{j}(q(x,y,z)) \land \overline{h}_{k}(s(x,y,z))$, and
$$
\ins_{j}' = \ins_{j-1} \cup \{\overline{h}_{j}(q(x,y,z)), \overline{h}_{k}(s(x,y,z))\} = \ins_{j+1}'.
$$
Thus, we have the derivation:
$$
\der_{3}' = \db, \ldots, (\rho_{3},h',\ins_{j}'), (\rho_{k},h_{k},\ins_{j+1}'), \ldots, (\rho_{n+1},h_{n+1},\ins_{n+1})
$$
It is straightforward to confirm that $\der_{3}'$ is indeed a valid derivation, and furthermore, for any rule $\rho_{m}$ with $j < m < n+1$, if it depended on $\rho_{j}$, it will now depend on the above application of $\rho_{3}$, showing that for any such $m$ we have
\begin{align*}
h_{m}(\fr(\rho_{m})) & \subseteq \nullset(\overline{h}_{j}(\head(\rho_{1}))) \cup \conset(\db,\rs)\\
& \subseteq \nullset(\overline{h'}(\head(\rho_{3}))) \cup \conset(\db,\rs).
\end{align*}
Moreover, $\rho_{n+1} = \rho_{4}$ can be seen to depend on the application of $\rho_{3}$ displayed in $\der_{3}'$ above, that is to say
\begin{align*}
h_{n+1}(\fr(\rho_{n+1})) & \subseteq \nullset(\overline{h'}(\head(\rho_{3}))) \cup \conset(\db,\rs).
\end{align*}
Hence, it follows that $\der_{3}'$ is greedy.

\item If $\rho_{j} = \rho_{3}$ and $\rho_{k} = \rho_{2}$, then replace $(\rho_{k},h_{k},\ins_{k})$ in $\der_{3}$ with $(\rho_{3},h',\ins_{j+1}')$ such that $h'(p(x)) = h_{j}(p(x))$, $h'(r(y)) = h_{k}(r(x))$, $\overline{h'}(\head(\rho_{3})) = \overline{h}_{j}(q(x,y,z)) \land \overline{h}_{k}(s(x,y,z))$, and
$$
\ins_{j+1}' = \ins_{j} \cup \{\overline{h}_{j}(q(x,y,z)), \overline{h}_{k}(s(x,y,z))\}.
$$
Thus, we have the derivation:
$$
\der_{3}' = \db, \ldots, (\rho_{j},h_{j},\ins_{j}), (\rho_{3},h',\ins_{j+1}'), \ldots, (\rho_{n+1},h_{n+1},\ins_{n+1})
$$
It is straightforward to confirm that $\der_{3}'$ is indeed a valid derivation, and furthermore, for any rule $\rho_{m}$ with $k < m < n+1$, if it depended on $\rho_{k}$, it will now depend on the above application of $\rho_{3}$, showing that for any such $m$ we have
\begin{align*}
h_{m}(\fr(\rho_{m})) & \subseteq \nullset(\overline{h}_{k}(\head(\rho_{2}))) \cup \conset(\db,\rs)\\
& \subseteq \nullset(\overline{h'}(\head(\rho_{3}))) \cup \conset(\db,\rs).
\end{align*}
Additionally, $\rho_{n+1} = \rho_{4}$ can be seen to depend on the application of $\rho_{3}$ displayed in $\der_{3}'$ above, that is to say
\begin{align*}
h_{n+1}(\fr(\rho_{n+1})) & \subseteq \nullset(\overline{h'}(\head(\rho_{3}))) \cup \conset(\db,\rs).
\end{align*}
Hence, it follows that $\der_{3}'$ is greedy.

\item If $\rho_{j} = \rho_{3}$ and $\rho_{k} = \rho_{3}$, then add $(\rho_{3},h',\ins_{j+1}')$ after the two inferences $(\rho_{j},h_{j},\ins_{j}),(\rho_{k},h_{k},\ins_{j+1}')$ in $\der_{3}$ where $h'(p(x)) = h_{j}(p(x))$, $h'(r(y)) = h_{k}(r(y))$, and $\overline{h'}(\head(\rho_{3})) = \overline{h}_{j}(q(x,y,z)) \land \overline{h}_{k}(s(x,y,z))$. Thus, we have the derivation: $\der_{3}' =$
$$
\db, \ldots, (\rho_{j},h_{j},\ins_{j}), (\rho_{k},h_{k},\ins_{j+1}'), (\rho_{3},h',\ins_{j+1}'), \ldots, (\rho_{n+1},h_{n+1},\ins_{n+1})
$$
It is straightforward to confirm that $\der_{3}'$ is indeed a valid derivation. Also, observe
$$
h'(\fr(\rho_{3}))  \subseteq \conset(\db,\rs) \subseteq \nullset(\overline{h}_{l}(\head(\rho_{l}))) \cup \conset(\db,\rs).
$$
for any $l \leq k$. Moreover, for every rule $\rho_{m}$ with $k < m < n+1$, if 
\begin{align*}
h_{m}(\fr(\rho_{m})) & \subseteq \nullset(\overline{h}_{m'}(\head(\rho_{m'}))) \cup \conset(\db,\rs).
\end{align*}
held in $\der_{3}$ with $m' < m$, then it will continue to hold in $\der_{3}'$. Last, $\rho_{n+1} = \rho_{4}$ can be seen to depend on the application of $\rho_{3}$ displayed in $\der_{3}'$ above, that is to say
\begin{align*}
h_{n+1}(\fr(\rho_{n+1})) & \subseteq \nullset(\overline{h'}(\head(\rho_{3}))) \cup \conset(\db,\rs).
\end{align*}
Hence, it follows that $\der_{3}'$ is greedy, and concludes our proof that $\rsii$ is a $\wgbts$, but is not a $\gbts$.
\end{enumerate}
\end{enumerate}
\end{proof}

\section{Proofs for \Cref{sec:derivation-graphs}}

\begin{lemma}\label{lem:basic-dg-properties}
Let $\db$ be a database and $\rs$ a rule set. If $\db \drel{\rs}{\der} \ins$ with $\os(\dg) = (\nodes, \edges, \at, \lbl)$ a (potentially reduced) derivation graph and $\os$ a reduction sequence, then $\os(\dg)$ has the following properties:
\begin{enumerate}


\item for each non-initial node $X_{\vn} \in \nodes$, there exists a $\rho \in \rs$ with $\rho = \phi(\vec{x},\vec{y}) \rightarrow \exists \vec{z} \psi(\vec{y},\vec{z})$ and a homomorphism $\overline{h}$ such that $\at(X_{\vn}) = \overline{h}(\psi(\vec{y},\vec{z}))$;

\item if $(X_{\vn},X_{\vm}) \in \edges$, then $\vn \ilt \vm$.

\end{enumerate}
\end{lemma}

\begin{proof} Both claims follow from the definition of a derivation graph along with the fact that the reduction operations only affect arcs and labels.
\end{proof}

\begin{definition}[$x$-Generative, Source Node]
Let $\db$ be a database, $\rs$ be a rule set, $\db \drel{\der}{\rs} \ins$, and $\os$ be a reduction sequence applicable to $\dg = (\nodes,\edges,\at,\lbl)$. We define a node in $\os(\dg) = (\nodes',\edges',\at',\lbl')$ to be \emph{$x$-generative} with $x \in \nonc(\ins)$ \iffi for every node $X_{k} \in \nodes'$, if $x \in \nonc(X_{k})$, then $n \leq k$. We define a node $X \in \nodes'$ to be a \emph{source node} \iffi no node $Y \in \nodes'$ exists such that $(Y,X) \in \edges'$, and we define $X$ to be \emph{non-source node} otherwise.
\end{definition}

\begin{lemma}\label{lem:existing-parent-node}
Let $\db$ be a database, $\rs$ be a rule set, $\db \drel{\der}{\rs} \ins$, and $\os$ be a reduction sequence with $\os(\dg) = (\nodes,\edges,\at,\lbl)$. For any nodes $X_{i}, X_{j} \in \nodes$, if $x \in \nonc(X_{i}) \cap \nonc(X_{j})$, $(X_{i},X_{j}) \in \edges$, and $x \not\in \lbl(X_{i},X_{j})$, then there exists a node $X_{m} \in \nodes$ such that $x \in \nonc(X_{m})$, $(X_{m},X_{j}) \in \edges$, and $x \in \lbl(X_{m},X_{j})$.
\end{lemma}


\begin{lemma}\label{lem:connect-t-gen}
Let $\db$ be a database and $\rs$ be a rule set. If $\db \drel{\der}{\rs} \ins$, then for any reduction sequence $\os$, $\os(\dg) = (\nodes,\edges,\at,\lbl)$ satisfies the following two conditions:
\begin{enumerate}

\item if $x \in \nonc(\ins)$, then there exists a unique $x$-generative node $X \in \nodes$;

\item if $X_{n}$ is the $x$-generative node in $\os(\dg)$, then for every $X_{k} \in \nodes$ such that $x \in \nonc(X_{k})$, there is a directed path from $X_{n}$ to $X_{k}$ in $\os(\dg)$ such that for every node $X_{\ell}$ along the path, $\ell \leq k$ and $x \in \nonc(X_{\ell})$.

\end{enumerate}
\end{lemma}

\begin{proof} Statement 1 is evident as there must be a first rule application in $\der$ that introduces the null $x$. Let $\os(\dg) = (\nodes,\edges,\at,\lbl)$. We argue statement 2 by induction on the lexicographic ordering of pairs $(|\der|,|\os|)$, where $|\der|$ is the length of the derivation and $|\os|$ is the length of the reduction sequence. Suppose $X_{n}$ is the $x$-generative node in $\os(\dg)$ and let $X_{k} \in \nodes$ such that $x \in \nonc(X_{k})$. We aim to show that a directed path exists from $X_{n}$ to $X_{k}$ such that for every node $X_{\ell}$ along the path, $\ell \leq k$ and $x \in \nonc(X_{\ell})$.

\textit{Base case.} If $|\der| = 0$, meaning $\der = \db$, then the result trivially follows. If $|\os| = 0$, then $\os(\dg) = \dg$ with $\dg = (\nodes,\edges,\at,\lbl)$. If $X_{k} = X_{n}$, then the claim trivially holds. However, if $X_{k} \neq X_{n}$, then let us consider the derivation $\db \drel{\der'}{\rs} \ins_{k-1}, (\rho,h,\ins_{k})$, where the application of $\rho$ produces the node $X_{k}$. Since $x \in \nonc(X_{n})$ and $x \in \nonc(X_{k})$, we know there exists a node $X_{m}$ such that $x \in \nonc(X_{m})$, $(X_{m},X_{k}) \in \edges$, and $x \in \lbl(X_{m},X_{k})$. By IH, there is a directed path from $X_{n}$ to $X_{m}$ such that for every $X_{\ell}$ along the path $\ell \leq m$ and $x \in \nonc(X_{\ell})$. Therefore, since $(X_{m},X_{k}) \in \edges$, we know that such a directed path from $X_{n}$ to $X_{k}$ of the required shape exists as well.

\textit{Inductive step.} Let $\ro \in \{\TR, \AR, \CR\}$ with $\os =\ro \os'$. Let $\os'(\dg) = (\nodes',\edges',\at',\lbl')$. We consider the cases where $\ro$ is either $\AR$ or $\CR$ as the case when $\ro$ is $\TR$ is trivial as all paths are preserved after the reduction operation is applied.

$\AR$. Suppose that an arc $(X_{i},X_{j}) \in \edges'$ exists such that $\lbl'(X_{i},X_{j}) = \emptyset$, which is removed by applying $\AR$ to $\os'(\dg)$. By IH, we know that a directed path from $X_{n}$ to $X_{k}$ exists in $\os'(\dg)$ such that for every node $X_{\ell}$ along the path $\ell \leq k$ and $x \in \nonc(X_{\ell})$. Let us suppose that $(X_{i},X_{j})$ occurs along this path, since otherwise, the result trivially follows. Then, we know that $x \in \nonc(X_{i})$ and $x \in \nonc(X_{j})$. Since $\lbl'(X_{i},X_{j}) = \emptyset$, we know $x \not\in \lbl'(X_{i},X_{j})$, and therefore, by \lem~\ref{lem:existing-parent-node}, some $X_{m} \neq X_{i}$ exists such that $X_{m} \in \nodes'$, $x \in \nonc(X_{m})$, $(X_{m},X_{j}) \in \edges'$, and $x \in \lbl'(X_{m},X_{j})$. By IH, there exists a directed path from $X_{n}$ to $X_{m}$ such that for every node $X_{\ell}$ along the path $\ell \leq m$ and $x \in \nonc(X_{\ell})$. After $\AR$ is applied, this path will still be present, and so, a path of the desired shape will exist from $X_{n}$ to $X_{k}$.

$\CR$. Suppose that $(X_{i},X_{m}), (X_{j},X_{m}) \in \edges'$, and there exists a node $X_{\ell}$ such that $\ell < m$ and $\lbl'(X_{i},X_{m}) \cup \lbl'(X_{j},X_{m}) \subseteq \termset(X_{\ell})$.  After applying $\CR$, we suppose that $(X_{i},X_{m}), (X_{j},X_{m})$ are removed from the set of arcs and $(X_{\ell},X_{m})$ is added such that $\lbl(X_{\ell},X_{m}) = \lbl'(X_{i},X_{m}) \cup \lbl'(X_{j},X_{m})$. By IH, a directed path from $X_{n}$ to $X_{k}$ exists in $\os'(\dg)$ such that for every node $X_{u}$ along the path $u \leq k$ and $x \in \nonc(X_{u})$. We assume w.l.o.g. that $(X_{i},X_{m})$ occurs along this path, since the other cases are trivial or similar. If $x \in \lbl'(X_{i},X_{m})$, then $x \in \nonc(X_{\ell})$ by assumption, implying that a directed path exists from $X_{n}$ to $X_{\ell}$ (because $X_{n}$ is assumed $x$-generative) of the required form. Hence, after applying $\CR$, a directed path of the required form will exist consisting of the path from $X_{n}$ to $X_{\ell}$, the arc $(X_{\ell},X_{m})$, and the path from $X_{m}$ to $X_{k}$. However, if $x \not\in \lbl'(X_{i},X_{m})$, then as in the $\AR$ case above, there exists some $X_{v} \neq X_{i}$ such that $X_{v} \in \nodes'$, $x \in \nonc(X_{v})$, $(X_{v},X_{m}) \in \edges'$, and $x \in \lbl'(X_{v},X_{m})$ (by \lem~\ref{lem:existing-parent-node}). By an argument similar to the $\AR$ case, we find that a directed path of the required form exists from $X_{n}$ to $X_{k}$ in $\os(\dg)$.
\end{proof}

\begin{customlem}{\ref{lem:reduct-ops-preserve-decomp-props}}
Let $\db$ be a database and $\rs$ be a rule set. If $\db \drel{\der}{\rs} \ins$, then for any reduction sequence $\os$, $\os(\dg) = (\nodes,\edges,\at,\lbl)$ satisfies the decomposition properties 1-4 in \lem~\ref{lem:decomposition-properties}, i.e. the following four conditions:
\begin{enumerate}
\item $\bigcup_{X_{\vn} \in \nodes} \termset(X_{\vn}) = \termset(\ins)$;
\item For each $p(\vec{t}) \in \ins$, there is an $X_{\vn} \in \nodes$ such that $p(\vec{t}) \in \at(X_{\vn})$;
\item For each term $x \in \nonc(\ins)$, the subgraph of $\os(\dg)$ induced by the nodes $X_{\vn}$ such that $x \in \nonc(X_{\vn})$ is connected;
\item For each $X_{\vn} \in \nodes$ the size of $\termset(X_{\vn})$ is bounded by an integer that only depends on the size of $(\db,\rs)$, \textit{viz.} $max\{|\termset(\db)|,|\termset(\head(\rho_{i}))|_{\rho_{i} \in \rs}\} + |\cons|$.
\end{enumerate}
\end{customlem}

\begin{proof} It is straightforward to confirm properties 1, 2, and 4. Property 3 follows from \lem~\ref{lem:connect-t-gen}.
\end{proof}

\begin{lemma}\label{lem:compactness} Let $\rs$ be a rule set. If for every database $\db$, there exists an $n \in \mathbb{N}$ such that for every $k \in \mathbb{N}$, $\tw{\chk{k}{\db,\rs}} \leq n$, then $\rs$ is $\sbts$.
\end{lemma}

\begin{proof} Let $\rs$ be a rule set such that for every database $\db$, there exists an $n \in \mathbb{N}$ such that for every $k \in \mathbb{N}$, $\tw{\chk{k}{\db,\rs}} \leq n$. Let $\db$ be an arbitrary database. As $\chk{k}{\db,\rs}$ is finite for every $k \in \mathbb{N}$ and monotonically increases (relative to the subset relation) as $k$ increases, we have that for every finite subset of $\ch_{\infty}(\db,\rs)$, the treewidth of that subset is bounded by $n$. Thus, by the the treewidth compactness theorem~\cite{Tho88}, $\tw{\ch_{\infty}(\db,\rs)} \leq n$. Since $\ch_{\infty}(\db,\rs)$ is a universal model of $(\db,\rs)$, it follows that $(\db,\rs)$ has a universal model of finite treewidth. Last, since $\db$ was assumed arbitrary, we have that $\rs$ is $\bts$.
\end{proof}

\begin{customprop}{\ref{prop:(w)adgs-relationships}} Every $\cdgs$ rule set is $\wcdgs$ and every $\wcdgs$ rule set is $\bts$.
\end{customprop}

\begin{proof} We prove each conjunct of the claim in turn:
\begin{enumerate}

\item Suppose that $\rs$ is $\cdgs$ and let $\db$ be an arbitrary database. Then, if $\db \drel{\rs}{\der} \ins$, it follows that a derivation $\der' = \der$ exists such that $\db \drel{\rs}{\der'} \ins$ and $G_{\der'}$ can be reduced to a cycle-free graph (since $\rs$ is $\cdgs$). Hence, $\rs$ is $\wcdgs$.

\item Suppose that $\rs$ is $\wcdgs$, $\db$ is a database, let $C = \conset(\db,\rs)$, and let $n = max\{|\termset(\db)|,|\termset(\head(\rho_{i}))|_{\rho_{i} \in \rs}\} + |\cons|$, and assume that $\db \drel{\der}{\rs} \ins$. Our first aim is to show that $\tw{\ins} \leq n$, thus showing that any $\rs$-derivable instance from $\db$ has a treewidth bounded by $n$. Since $\rs$ is $\wcdgs$, we know there exists an $\rs$-derivation $\der'$ and a complete reduction sequence $\os$ such that $\os(G_{\der'}) = (\nodes',\edges',\at',\lbl')$ is a cycle-free graph. Let us define a tree decomposition $\td = (\tdv,\tde)$ of $\ins$ by making use of $\os(G_{\der'})$, where $X \in \tdv$ \iffi there exists a node $X' \in \nodes'$ such that $X = \termset(X')$. We then define $(X,Y) \in \tde''$ \iffi there exists an arc $(X',Y') \in \edges'$ such that $X = \termset(X')$ and $Y = \termset(Y')$. In general, $\td' = (\tdv,\tde'')$ will be a finite forest, so if we place each tree of $\td'$ in a line and connect the root of the first tree to the root of the second, the root of the second tree to the root of the third, etc., then this yields a tree decomposition $\td = (\tdv,\tde)$ (where $\tde$ extends $\tde''$ with the edges just mentioned). By \lem~\ref{lem:reduct-ops-preserve-decomp-props}, $\td$ is indeed a tree decomposition, and furthermore, $\w{\td} \leq n$. Thus, we have shown that any $\rs$-derivable instance from $\db$ has a treewidth bounded by $n$. Now, observe that $\db \drel{\der}{\rs} \chk{k}{\db,\rs}$ for every $k \in \mathbb{N}$. Therefore, since $\db$ was assumed arbitrary, we know that for every database $\db$, there exists an $n \in \mathbb{N}$ such that for every $k \in \mathbb{N}$, $\tw{\chk{k}{\db,\rs}} \leq n$. By \lem~\ref{lem:compactness} above, it follows that $\rs$ is $\bts$, establishing the claim.
\end{enumerate}
\end{proof}

\begin{definition}\label{def:fr-nul-con-node} Let $\os(\dg) = (\nodes, \edges, \at, \lbl)$ be a derivation graph with $\os$ a reduction sequence and $X_{\vn} \in \nodes$. Moreover, let $\db$ be a database, $\rs$ be a rule set, and $C = \conset(\db,\rs)$. We define the \emph{frontier} $\fr(X_{\vn})$ of a node $X_{\vn} \in \nodes$ relative to $(\db,\rs)$ accordingly:
\[
  \fr(X_{n}) =
  \begin{cases}
     \emptyset & \text{if $X_{n}$ is a source node;} \\
     \overline{h}_{i}(\vec{y}_{i}) \setminus C & \text{otherwise.}
  \end{cases}
\]
where $\at(X_{n}) = \overline{h}_{i}(\psi_{i}(\vec{y}_{i},\vec{z}_{i}))$.
\end{definition}

\begin{lemma}\label{lem:derivation-graph-reduction-properties}
Let $\db$ be a database, $\rs$ be a rule set, $C = \conset(\db,\rs)$, and assume that $\db \drel{\rs}{\der} \ins$. Then, for $\os$ a reduction sequence, the derivation graph $\os(\dg) = (\nodes, \edges, \at, \lbl)$ satsifies the following properties:
\begin{enumerate}

\item for each $X_{\vn_{0}} \in \nodes$ with parent nodes $X_{\vn_{1}}, \ldots, X_{\vn_{k}} \in \nodes$, 
$$
\fr(X_{\vn_{0}}) = \bigcup_{i \in \{1,\ldots, k\}} \lbl(X_{\vn_{i}},X_{\vn_{0}});
$$

\item for each $(X_{\vm},X_{\vn}) \in \edges$, $\lbl(X_{\vm},X_{\vn}) \subseteq \termset(X_{\vm})$;

\item for each $X_{\vn_{0}} \in \nodes$ with parent nodes $X_{\vn_{1}}, \ldots, X_{\vn_{k}} \in \nodes$, 
$$
\bigcup_{i \in \{1,\ldots, k\}} \lbl(X_{\vn_{i}},X_{\vn_{0}}) \subseteq \bigcup_{i \in \{1,\ldots, k\}} \termset(X_{\vn_{i}}).
$$

\end{enumerate}
\end{lemma}

\begin{proof} Since 3 follows from 2, we only prove 1 and 2. We prove each claim in turn by induction on the length of the reduction sequence $\os$.

\begin{enumerate}

\item \textit{Base case.} Suppose that $\os = \empstr$, so that $\os(\dg) = \empstr(\dg) = \dg$. Observe that for any $X_{n} \in \nodes$ with (a non-empty set of) parent nodes $X_{n_{1}}, \ldots, X_{n_{k}} \in \nodes$, $\fr(X_{n}) = \overline{h}(\vec{y}) \setminus C$ for $\psi(\vec{y},\vec{z}) = \head(\rho)$ for some $\rho \in \rs$. Moreover, by definition, it follows that
$$
\fr(X_{n}) = \bigcup_{i \in \{1,\ldots, k\}} \lbl(X_{n_{i}},X_{n}).
$$

\textit{Inductive step.} We assume for IH that the property holds for $\os(\dg)$ and show that the property holds for $\RO \os(\dg) = (\nodes',\edges',\at',\lbl')$ with $\RO \in \{\TR, \AR, \CR\}$. We make a case distinction based on the last reduction operation $\RO$ applied.

\medskip

$\AR$. Let $(X_{\vn_{1}},X_{\vn_{0}}) \in \edges$ such that $\lbl(X_{\vn_{1}},X_{\vn_{0}}) = \emptyset$. Assume that $\AR$ was applied, so that $(X_{\vn_{1}},X_{\vn_{0}}) \not\in \edges'$. For any node $X_{\vm} \neq X_{\vn_{0}}$ in $\RO\os(\dg)$ property 2 holds by IH, and for the node $X_{\vn_{0}}$ with parent nodes $X_{\vm_{1}}, \ldots, X_{\vm_{k}}$ in $\AR\os(\dg)$ we have
\begin{align*}
\fr(X_{\vn_{0}}) & = \bigcup_{i \in \{1,\ldots, k\}} \lbl(X_{\vm_{i}},X_{\vn_{0}}) \cup \lbl(X_{\vn_{1}},X_{\vn_{0}})\\
& = \bigcup_{i \in \{1,\ldots, k\}} \lbl'(X_{\vm_{i}},X_{\vn_{0}}) \cup \emptyset \\
& = \bigcup_{i \in \{1,\ldots, k\}} \lbl'(X_{\vm_{i}},X_{\vn_{0}})
\end{align*}
where the first equality follows by IH, and the second by the definition of $\lbl'$ along with the fact that $\lbl(X_{\vn_{1}},X_{\vn_{0}}) = \emptyset$.

\medskip

$\TR$. Let $(X_{\vn_{1}},X_{\vn_{0}}), (X_{\vn_{2}},X_{\vn_{0}}) \in \edges$ with $t \in \lbl(X_{\vn_{1}},X_{\vn_{0}}) \cap \lbl(X_{\vn_{2}},X_{\vn_{0}})$. Suppose we apply $\TR$, so that $\lbl'(X_{\vn_{2}},X_{\vn_{0}}) = \lbl(X_{\vn_{2}},X_{\vn_{0}}) \setminus \{t\}$. For any node $X_{\vm} \neq X_{\vn_{0}}$ in $\TR\os(\dg)$ the result holds by IH, and for the node $X_{\vn_{0}}$ with parent nodes $X_{\vm_{1}}, \ldots, X_{\vm_{k}}$ we have
$$
\fr(X_{\vn_{0}}) = \bigcup_{i \in \{1,\ldots, k\}} \lbl(X_{\vm_{i}},X_{\vn_{0}}) = \bigcup_{i \in \{1,\ldots, k\}} \lbl'(X_{\vm_{i}},X_{\vn_{0}})
$$
as $t \in \lbl'(X_{\vn_{1}},X_{\vn_{0}})$.

\medskip

$\CR$. Let $(X_{\vn_{1}},X_{\vn_{0}}), (X_{\vn_{2}},X_{\vn_{0}}) \in \edges$ with a node $X_{\vm} \in \nodes$ such that $\vm \ilt \vn_{0}$ and $\lbl(X_{\vn_{1}},X_{\vn_{0}}) \cup \lbl(X_{\vn_{2}},X_{\vn_{0}}) \subseteq \termset(X_{\vm})$. Assume that $\CR$ was applied, so that $(X_{\vm},X_{\vn_{0}}) \in \edges'$ with $\lbl'(X_{\vm},X_{\vn_{0}})  = \lbl(X_{\vn_{1}},X_{\vn_{0}}) \cup \lbl(X_{\vn_{2}},X_{\vn_{0}})$. For any node $X_{\vk} \neq X_{\vn_{0}}$ in $\CR\os(\dg)$, property 2 holds by IH, so let us consider the node $X_{\vn_{0}}$, which has parents $X_{\vm_{1}}, \ldots, X_{\vm_{k}}$, $X_{\vn_{1}}$, and $X_{\vn_{2}}$ in $\os(\dg)$ and parents $X_{\vm_{1}}, \ldots, X_{\vm_{k}}$, and $X_{\vm}$ in $\CR\os(\dg)$. By IH, we have the first equality below, and the second follows from the definition of $\lbl'$, giving the desired result:
\begin{align*}
\fr(X_{\vn_{0}}) & = \bigcup_{i \in \{1,\ldots, k\}} \lbl(X_{\vm_{i}},X_{\vn_{0}}) \cup \lbl(X_{\vn_{1}},X_{\vn_{0}}) \cup \lbl(X_{\vn_{2}},X_{\vn_{0}})\\
& = \bigcup_{i \in \{1,\ldots, k\}} \lbl'(X_{\vm_{i}},X_{\vn_{0}}) \cup \lbl'(X_{\vm},X_{\vn_{0}}).
\end{align*}

\item \textit{Base case.} Suppose that $\os = \empstr$, so that $\os(\dg) = \empstr(\dg) = \dg$. The result immediately follows from the definition of an derivation graph.

\medskip

\textit{Inductive step.} We assume for IH that the property holds for $\os(\dg)$ and show that the property holds for $\RO \os(\dg) = (\nodes',\edges',\at',\lbl')$ with $\RO \in \{\TR, \AR, \CR\}$.

\medskip

$\TR$. Let $(X_{\vn_{1}},X_{\vn_{0}}), (X_{\vn_{2}},X_{\vn_{0}}) \in \edges$ with $t \in \lbl(X_{\vn_{1}},X_{\vn_{0}}) \cap \lbl(X_{\vn_{2}},X_{\vn_{0}})$. Suppose we apply $\TR$, so that $\lbl'(X_{\vn_{2}},X_{\vn_{0}}) = \lbl(X_{\vn_{2}},X_{\vn_{0}}) \setminus \{t\}$. For any arc $(X_{\vm_{1}},X_{\vm_{0}}) \neq (X_{\vn_{2}},X_{\vn_{0}})$ in $\TR\os(\dg)$, the result holds by IH, so let us focus on $(X_{\vn_{2}},X_{\vn_{0}}) \in \edges'$. Observe that $\lbl'(X_{\vn_{2}},X_{\vn_{0}}) \subseteq \lbl(X_{\vn_{2}},X_{\vn_{0}}) \subseteq \termset(X_{\vn_{2}})$.

\medskip

$\AR$. Let $(X_{\vn_{1}},X_{\vn_{0}}) \in \edges$ such that $\lbl(X_{\vn_{1}},X_{\vn_{0}}) = \emptyset$. Assume that $\AR$ was applied, so that $(X_{\vn_{1}},X_{\vn_{0}}) \not\in \edges'$. For any $(X_{\vm_{1}},X_{\vm_{0}}) \in \edges'$, $\lbl'(X_{\vm_{1}},X_{\vm_{0}}) = \lbl(X_{\vm_{1}},X_{\vm_{0}}) \subseteq \termset(X_{\vm_{1}})$ by the definition of $\lbl'$ and IH. 

\medskip

$\CR$. Let $(X_{\vn_{1}},X_{\vn_{0}}), (X_{\vn_{1}},X_{\vn_{0}}) \in \edges$ with a node $X_{\vm} \in \nodes$ such that $\vm \ilt \vn_{0}$ and $\lbl(X_{\vn_{1}},X_{\vn_{0}}) \cup \lbl(X_{\vn_{1}},X_{\vn_{0}}) \subseteq \termset(X_{\vm})$. Assume that $\CR$ was applied, so that $(X_{\vm},X_{\vn_{0}}) \in \edges'$ with $\lbl'(X_{\vm},X_{\vn_{0}})  = \lbl(X_{\vn_{1}},X_{\vn_{0}}) \cup \lbl(X_{\vn_{1}},X_{\vn_{0}})$. For any arc $(X_{\vk_{1}},X_{\vk_{2}}) \neq (X_{\vm},X_{\vn_{0}})$ in $\CR\os(\dg)$, the result holds by IH, so let us focus on $(X_{\vm},X_{\vn_{0}}) \in \edges$. We have $\lbl'(X_{\vm},X_{\vn_{0}}) = \lbl(X_{\vn_{1}},X_{\vn_{0}}) \cup \lbl(X_{\vn_{1}},X_{\vn_{0}}) \subseteq \termset(X_{\vm})$ by the definition of $\lbl'$ and the condition required to apply $\CR$. This concludes the proof of the case.
\end{enumerate}
\end{proof}

\begin{definition}[Sub-reduction Sequence] Let $\os = \mathsf{(r_{1})} \cdots \mathsf{(r_{n})}$ be a reduction sequence. We define a \emph{sub-reduction sequence} $\os'$ of $\os$ to be a reduction sequence of the form $\mathsf{(r_{1})} \cdots \mathsf{(r_{i})}$ with $0 \leq i \leq n$, which is the empty reduction sequence $\empstr$ when $n = 0$. If $\os'$ is a sub-reduction sequence of $\os$, then we write $\os' \subos \os$, and we note that we take $\os'$ to be the same instances of the reduction operations occurring within the reduction sequence $\os$.
\end{definition}

\begin{lemma}\label{lem:reverse-preservation-of-derivation-graph-properties}
Let $\db$ be a database, $\rs$ be a rule set, and assume that $\db \drel{\rs}{\der} \ins$. Moreover, assume that $\os(\dg) = (\nodes, \edges, \at, \lbl)$ is a cycle-free derivation graph with $\os$ a complete reduction sequence. For each $\os' \subos \os$, the derivation graph $\os'(\dg) = (\nodes',\edges',\at',\lbl')$ satisfies the following: For each non-source node $X_{\vn} \in \nodes'$, there exists a node $X_{\vm} \in \nodes'$ such that $\vm \ilt \vn$ and $\fr(X_{\vn}) \subseteq \termset(X_{\vm})$.
\end{lemma}

\begin{proof} We first show (1) that the claim holds for $\os(\dg)$, and then (2) show that if the claim holds for $\os'(\dg)$ with $\os' = \RO \os''$ and $\RO \in \{\TR,\AR,\CR\}$, then it holds for $\os''(\dg)$.

(1) Let $X_{\vn} \in \nodes$ be a non-source node of $\os(\dg)$ with parent nodes $X_{\vn_{i}}$ for $i \in \{1, \ldots, k\}$. By \lem~\ref{lem:derivation-graph-reduction-properties}, we know that
$$
\fr(X_{\vn}) = \bigcup_{i \in \{1,\ldots, k\}} \lbl(X_{\vn_{i}},X_{\N}) \subseteq \bigcup_{i \in \{1,\ldots, k\}} \termset(X_{\vn_{i}}).
$$
Since $\os$ is a complete reduction sequence and $\os(\dg)$ is cycle-free, we know that $\os(\dg)$ is a forest, implying that each non-source node has a single parent node. Hence, $X_{\vn}$ has a single parent node $X_{\vm}$, implying that $\fr(X_{\vn}) \subseteq \termset(X_{\vm})$, 
thus confirming the desired result as $\vm \ilt \vn$ by \lem~\ref{lem:basic-dg-properties}. 

(2) Let $\os'(\dg) = (\nodes',\edges',\at',\lbl')$, $\os''(\dg) = (\nodes'',\edges'',\at'',\lbl'')$, and suppose that for every non-source node $X_{\vn} \in \nodes'$, there exists a node $X_{\vm} \in \nodes'$ such that $\vm \ilt \vn$ and $\fr(X_{\vn}) \subseteq \termset(X_{\vm})$. We show the claim by a case distinction on if $\TR$, $\AR$, or $\CR$ was applied last in $\os'$.

$\TR$. Observe that if $\TR$ was applied last in $\os'$, then the only difference between $\os'(\dg)$ and $\os''(\dg)$ is that for some arc $(X_{\vk_{1}},X_{\vk_{0}}) \in \edges' \cap \edges''$, $\lbl''(X_{\vk_{1}},X_{\vk_{0}}) = \lbl'(X_{\vk_{1}},X_{\vk_{0}}) \cup \{t\}$ , for some term $t$. Hence, for an arbitrary non-source node $X_{\vn} \in \nodes''$, $X_{\vn} \in \nodes'$ since $\nodes'' = \nodes'$, implying that there exists a node $X_{\vm} \in \nodes' = \nodes''$ such that $\vm \ilt \vn$ and $\fr(X_{\vn}) \subseteq \termset(X_{\vm})$, completing the proof of the case.

$\AR$. If $\AR$ was applied last in $\os'$, then the only difference between $\os'(\dg)$ and $\os''(\dg)$ is that for some arc $(X_{\vk_{1}},X_{\vk_{0}})$, $\edges'' = \edges' \cup \{(X_{\vk_{1}},X_{\vk_{0}})\}$, where $\lbl''(X_{\vk_{1}},X_{\vk_{0}}) = \emptyset$. For any non-source node $X_{\vn} \in \nodes''$ such that $X_{\vn}$ is a non-source node in $\nodes'$, the result immediately holds. However, it could be the case that even though $X_{\vk_{0}}$ is a non-source node in $\nodes''$, $X_{\vk_{0}}$ is a source node in $\os'(\dg)$ as $(X_{\vk_{1}},X_{\vk_{0}}) \in \edges''$. In this case, by \lem~\ref{lem:derivation-graph-reduction-properties} and the fact that $X_{\vk_{1}}$ is the only parent of $X_{\vk_{0}} \in \nodes''$, we know that $\fr(X_{\vk_{0}}) \subseteq \lbl''(X_{\vk_{1}},X_{\vk_{0}}) = \emptyset$, implying that $\fr(X_{\vk_{0}}) = \emptyset$. As $X_{\vk_{1}}$ is a parent of $X_{\vk_{0}}$ in $\os''(\dg)$, we know that $\vk_{1} \ilt \vk_{0}$ by \lem~\ref{lem:basic-dg-properties}, and trivially $\fr(X_{\vk_{0}}) \subseteq \termset(X_{\vk_{1}})$, proving the case.

$\CR$. If $\CR$ is applied last in $\os'$, then the only difference between $\os''(\dg)$ and $\os'(\dg)$ is that there exist arcs $(X_{n_{1}},X_{n_{0}}), (X_{n_{2}},X_{n_{0}}) \in \edges''$ and $\edges' = (\edges'' \setminus \{(X_{n_{1}},X_{n_{0}}), (X_{n_{2}},X_{n_{0}})\}) \cup \{(X_{m},X_{n_{0}})\}$ as there exists a node $X_{m} \in \nodes''$ such that $m \ilt n_{0}$ and $\lbl''(X_{n_{1}},X_{n_{0}}) \cup \lbl''(X_{n_{2}},X_{n_{0}}) = \termset(X_{m})$. Hence, for an arbitrary non-source node $X_{k} \in \nodes''$, $X_{k} \in \nodes'$ as $\nodes'' = \nodes'$, implying the existence of a node $X_{k'}$ such that $k' \ilt k$ and $\fr(X_{k}) \subseteq \termset(X_{k'})$, thus completing the proof.
\end{proof}

\begin{lemma}\label{lem:(w)gbts-is-(w)adgs}
Let $\rs$ be a rule set. Then,
\begin{enumerate}

\item If $\rs$ is $\gbts$, then $\rs$ is $\cdgs$;

\item if $\rs$ is $\wgbts$, then $\rs$ is $\wcdgs$.

\end{enumerate}
\end{lemma}

\begin{proof} We argue claim 1 since the proof of claim 2 is similar. Let $\db$ be a database, $\rs$ be $\gbts$, and assume $\db \drel{\der}{\rs} \ins$. Since $\rs$ is $\gbts$, we know that the $\rs$-derivation
$$
\der = \db,(\rho_{1},h_{1},\ins_{1}), \ldots, (\rho_{n},h_{n},\ins_{n})
$$
is greedy. Therefore, for each $i$ such that $0 < i < n$, there exists a $j < i$ such that $h_{i}(\fr(\rho_{i})) \subseteq \nullset(\overline{h}_{j}(\head(\rho_{j}))) \cup \conset(\db,\rs)$. Let us now show that $\rs$ is $\cdgs$ by arguing that $\dg = (\nodes,\edges,\at,\lbl)$ is reducible to a cycle-free graph. 

Let us suppose that there exist arcs $(X_{\vn_{1}},X_{\vn_{0}}), (X_{\vn_{2}},X_{\vn_{0}}) \in \edges$. By our assumption that $\der$ is greedy, we know that there exists a node $X_{\vm} \in \nodes$ such that $m < \vn_{0}$ and $\fr(X_{\vn_{0}}) \subseteq \termset(X_{\vm})$. By \lem~\ref{lem:derivation-graph-reduction-properties}, it follows that $\lbl(X_{\vn_{1}},X_{\vn_{0}}) \cup \lbl(X_{\vn_{2}},X_{\vn_{0}}) \subseteq \termset(X_{\vm})$, meaning we can apply $\CR$ to $\dg$. Observe that $\CR(\dg)$ has one less ``convergence point'' as $(X_{\vn_{1}},X_{\vn_{0}})$ and $(X_{\vn_{2}},X_{\vn_{0}})$ have been replaced by the single arc $(X_{\vm},X_{\vn_{0}})$. By repeating this process, all such convergence points will be removed, yielding a reduced, cycle-free derivation graph. Hence, $\rs$ is $\cdgs$.
\end{proof}

\begin{lemma}\label{lem:(w)adgs-is-(w)gbts}
Let $\rs$ be a rule set. Then,
\begin{enumerate}

\item If $\rs$ is $\cdgs$, then $\rs$ is $\gbts$;

\item if $\rs$ is $\wcdgs$, then $\rs$ is $\wgbts$.

\end{enumerate}
\end{lemma}

\begin{proof} We prove claim 2 since claim 1 is shown in a similar fashion. Let $\db$ be a database, $\rs$ be $\wcdgs$, $C = \conset(\db,\rs)$, and assume $\db \drel{\der}{\rs} \ins$. Since $\rs$ is $\wcdgs$, we know there exists an $\rs$-derivation 
$$
\der' = \db,(\rho_{1},h_{1},\ins_{1}), \ldots, (\rho_{k},h_{k},\ins_{k})
$$ 
such that $G_{\der'} = (\nodes,\edges,\at,\lbl)$ is reducible to a cycle-free graph. That is to say, there exists a (complete) reduction sequence $\os$ such that $\os(G_{\der'})$ is cycle-free. By \lem~\ref{lem:reverse-preservation-of-derivation-graph-properties}, we know that for every $\os' \subos \os$, $\os'(G_{\der'}) = (\nodes',\edges',\at',\lbl')$ satisfies the following property: For each non-source node $X_{\vn} \in \nodes'$, there exists a node $X_{\vm} \in \nodes'$ such that  $m < n$ and $\fr(X_{\vn}) \subseteq \termset(X_{\vm})$. In particular, this property holds for $\os' = \empstr$, i.e. for $G_{\der'}$. Since $h_{k}(\fr(\rho_{k})) \subseteq \cons$ when $X_{k} \in \nodes$ is a source node, and due to the fact that for each non-source node $X_{\vn} \in \nodes$, $\fr(X_{\vn}) = \overline{h}_{\vn}(\fr(\rho_{\vn})) \setminus C = h_{\vn}(\fr(\rho_{\vn})) \setminus C$, we have that
$$
h_{\vn}(\fr(\rho_{\vn})) \subseteq \fr(X_{\vn}) \cup C \subseteq \termset(X_{\vm}) = \nullset(\overline{h}_{\vm}(\head(\rho_{\vm})) \cup \conset(\db,\rs),
$$
 for each $0 < \vn \leq k$ and some $\vm < \vn$, where the last equality above follows from the definition of $\termset(X_{\vm})$. Therefore, $\der'$ is greedy, showing that $\rs$ is $\wgbts$.
\end{proof}

\begin{customthm}{\ref{thm:(w)gbts-(w)adgs-equivalence}}
$\gbts$ coincides with $\cdgs$ and $\wgbts$ coincides with $\wcdgs$. Membership in $\cdgs$, $\gbts$, $\wcdgs$, or $\wgbts$ warrants decidable BCQ entailment.
\end{customthm}

\begin{proof}
The first statement follows from \lem~\ref{lem:(w)gbts-is-(w)adgs} and \lem~\ref{lem:(w)adgs-is-(w)gbts}.
The second statement follows from the fact that BCQ entailment is decidable for $\bts$ and every class of rule sets mentioned is a subset of $\bts$.
\end{proof}

\begin{customcor}{\ref{corollaries}}
$\{\TR,\AR\}$ is reduction-admissible relative to $\CR$.
\end{customcor}

\begin{proof} 
Let $\der$ be an arbitrary $\rs$-derivation and assume that $\dg$ can be reduced to a cycle-free graph. By the proof of \lem~\ref{lem:(w)adgs-is-(w)gbts} above, $\der$ is a greedy $\rs$-derivation. Thus, by the proof of \lem~\ref{lem:(w)gbts-is-(w)adgs}, $\der$ is reducible to a cycle-free graph using only the $\CR$ operation.
\end{proof}

\end{document}